\documentclass[12pt]{amsart}
\usepackage{hyperref}
\usepackage{amsmath}
\usepackage{enumerate}
\usepackage{amsfonts}
\usepackage{amssymb}
\usepackage{latexsym}
\usepackage{fixmath} 
\usepackage{mathabx}
\usepackage{mathtools}
\usepackage{appendix}
\usepackage{placeins}
\usepackage{float}
\usepackage{bbm}
\usepackage{fullpage}
\usepackage[left=2.5cm,right=2.5cm,top=2.5cm,bottom=2.5cm]{geometry}
\usepackage[usenames,dvipsnames]{color}
\usepackage{graphicx}
\usepackage{tikz}
\usepackage[backend=biber, style=ieee]{biblatex}
\usepackage[linesnumbered,ruled,vlined]{algorithm2e}
\usepackage{algpseudocode}
\usepackage[foot]{amsaddr}

\renewcommand{\vec}[1]{\boldsymbol{#1}}

\newcommand\vz{\vec{z}}

\newcommand\eps{\epsilon}

\newcommand\zket{|0\rangle}
\newcommand\oket{|1\rangle}

\newcommand{\ket}[1]{ \vert #1 \rangle }

\newcommand\eul{\mathrm{e}}

\newcommand\chiket[1]{|\chi_{#1}\rangle}
\newcommand{\bra}[1]{\langle #1 \vert}

\newcommand\braphi{\langle \Phi}

\newcommand{\abs}[1]{\left\vert #1\right\vert}

\newcommand\ZZ{\mathbb{Z}}
\makeatletter
\renewcommand{\boxed}[1]{\text{\fboxsep=.2em\fbox{\m@th$\displaystyle#1$}}}
\makeatother
\makeatletter
\newcommand*\bigcdot{\mathpalette\bigcdot@{.5}}
\newcommand*\bigcdot@[2]{\mathbin{\vcenter{\hbox{\scalebox{#2}{$\m@th#1\bullet$}}}}}
\makeatother

\newcommand\JFsq{J_{\Fsq}}
\newcommand\fsq{f_{\rm sqr}}
\newcommand\Fsq{F_{\rm sqr}}

\newcommand\fmred{f^m_{\rm sqr}}

\newcommand\tp{\otimes}
\newcommand\parder[2]{\frac{\partial {#1}}{\partial {#2}}}

\newcommand\bc[1]{\left({#1}\right)}
\newcommand\cbc[1]{\left\{{#1}\right\}}
\newcommand\brk[1]{\left[{#1}\right]}

\newcommand\CC{\mathbb{C}}

\newcommand{\whp}{w.h.p.}

\newtheorem{definition}{Definition}[section]
\newtheorem{claim}[definition]{Claim}
\newtheorem{example}[definition]{Example}
\newtheorem{remark}[definition]{Remark}
\newtheorem{theorem}[definition]{Theorem}
\newtheorem{lemma}[definition]{Lemma}
\newtheorem{proposition}[definition]{Proposition}
\newtheorem{corollary}[definition]{Corollary}

\newcommand\prodsat{\mathrm{PRODSAT}~}

\newcommand\qsat{\mathrm{QSAT}}
\newcommand\lr{\mathrm{LR}}
 
\newcommand\multideg{\operatorname{multideg}}

\newcommand\dimKerH{\dim{\ker H_F}}
\newcommand\dimProd{\dim{\rm PRODSAT}}
\newcommand\alphaprod{\alpha_{\rm PRODSAT}}
\newcommand\alphalr{\alpha_{\rm lr}}
\newcommand\alphadc{\alpha_{\rm dc}}

\title{The PRODSAT phase of random quantum satisfiability}
\author{Joon Lee $^1$ }
\address{1 Leiden University - Leiden Institute of Advanced Computer Science}

\author{Nicolas Macris $^2$}
\address{2 EPFL - School of Computer and Communications Sciences}

\author{Jean Bernoulli Ravelomanana $^2$}
\author{Perrine Vantalon $^2$}

\thanks{j.h.lee@liacs.leidenuniv.nl, nicolas.macris@epfl.ch, rjeanbernoulli@gmail.com, perrine.vantalon@epfl.ch}

\begin{document}

\maketitle

\begin{abstract}
	The $k$-QSAT problem is a quantum analog of the famous $k$-SAT constraint satisfaction problem. We must determine the zero energy ground states of a Hamiltonian of $N$ qubits consisting of a sum of $M$ random $k$-local rank-one projectors. It is known that product states of zero energy exist with high probability if and only if the underlying factor graph has a clause-covering dimer configuration. This means that the  threshold of the PRODSAT phase is a purely geometric quantity equal to the dimer covering threshold. We revisit and fully prove this result through a combination of complex analysis and algebraic methods based on Buchberger's algorithm for complex polynomial equations with random coefficients. We also discuss numerical experiments investigating the presence of entanglement in the PRODSAT phase in the sense that product states do not span the whole zero energy ground state space.
\end{abstract}

\section{Introduction}
The quantum version of the $k$-satisfiability problem ($k$-QSAT) was introduced by Bravyi \cite{Bravyi} as a natural quantum analog of the classical $k$-SAT constraint satisfaction problem which has been a focal point in classical computer science ever since it was recognized to be NP-complete by Cook and Levin \cite{Cook, LevinTranslated}. 
In classical $k$-SAT one must determine the satisfiability of a logical formula in conjunctive normal form. Boolean variables $(x_1, \dots, x_N)\in \{0,1\}^N$ must simultaneously satisfy $M$ constraints in the form of disjunctions  of $k$ literals,
$\chi_{m_1} \lor \dots \lor \chi_{m_k}$ where $(m_1, m_2, \dots , m_k)\in \{1,\dots, N\}^k$ and  $\chi_{m_i}\in \{x_{m_i}, \neg x_{m_i}\}$. We label the constraints as $m\in \{1, \dots, M\}$ and by a slight abuse of notation identify $ m\equiv (m_1, \dots, m_k)$.
A $k$-SAT formula is a conjunction of disjunctions $F = \land_{m=1}^M(\lor_{i=1}^k\chi_{m_i})$. The satisfiability of $F$ may be formulated as the study of the set of zero energy assignments for the classical Hamiltonian function
\begin{align}\label{classical_cost}
h_F(x_1, \dots, x_{N}) = \sum_{m=1}^M \bc{1-\mathbbm{1}(\chi_{m_1} \lor \dots \lor \chi_{m_k})} 
\end{align}
In the quantum analog $k$-QSAT, the Boolean variables are replaced by $N$ quantum bits (or qubits) labeled $\{1, \dots, N\}$ in a collective pure state vector
$\ket{\Psi}$ belonging to the Hilbert space $\mathbb{C}^2_1\otimes\dots\otimes\mathbb{C}^2_N$.\footnote{We note that in this quantum model the degrees of freedom are {\it distinguishable} and hence labeled.} The state of the qubits must simultaneously satisfy a set of $M$ quantum constraints labeled by $m\in \{1,\dots, M\}$. Each constraint ensures that $\vert\Psi\rangle$ is in the kernel of the projector 
\begin{align}\label{m_projector}
\Pi_m= \ket{\Phi^m}\bra{\Phi^m}\otimes I_{N-k}
\end{align}
where $\ket{\Phi^m}\in \mathbb{C}^2_{m_1}\otimes\dots\otimes \mathbb{C}^2_{m_k}$ is a pure state vector of the Hilbert space of the $k$ qubits $(m_1, \dots, m_k)$ involved in the constraint $m$. The matrix $I_{N-k}$ is the $2^{N-k}\times 2^{N-k}$ identity matrix acting trivially on the remaining $N-k$ qubits not involved in the constraint $m$. A $k$-QSAT ``formula" $F$
is defined by the collection of projectors 
$\{\Pi_m, m\equiv(m_1, \dots, m_k)\}$.
In $k$-QSAT, the Hamiltonian is given by the $2^N\times 2^N$ matrix 
\begin{align}\label{quantum_hamiltonian}
H_F= \sum_{m=1}^M \Pi_m .
\end{align}
We say that $F$ is satisfiable if the Hamiltonian has zero energy eigenstates, in other words if $\ker(H_F)$ (also called kernel or ground state space) contains a non-trivial vector.

This Hamiltonian represents a natural quantum generalization of the classical cost function (\ref{classical_cost}). First note that each classical disjunction excludes one among $2^k$ assignments of its $k$ Boolean variables. Analogously each projector \eqref{m_projector} excludes one direction in the $2^k$ dimensional Hilbert space of $k$ qubits. Furthermore, it is not difficult to see that when
$\vert \Phi^m\rangle$'s are tensor products of computational basis vectors of $\mathbb{C}^2$, the matrix Hamiltonian \eqref{quantum_hamiltonian} reduces to a diagonal matrix (in the computational basis) with diagonal given by values of \eqref{classical_cost} for all possible $2^N$ assignments. Of course this is not so anymore when $\vert\Phi^m\rangle$ is an arbitrary state of $\mathbb{C}^2_{m_1}\otimes \dots\otimes \mathbb{C}^2_{m_k}$, be it a tensor product or entangled state. Finally it is well known that $k$-SAT is NP-complete for $k \geq 3$. The analogous statement for $k$-QSAT is that it is QMA-complete for $k \geq 3$ (see \cite{Bravyi,gosset2016quantum} for a discussion). 

In this paper we look at the {\it average case} analysis of $k$-QSAT. In this formulation the Hamiltonian is taken at random from a set of instances and the problem is to determine the typical behavior of the kernel space. This is perfectly analogous to the random $k$-SAT problem which studies the typical behavior of the space of zero cost assignments of a random formula.
To formulate the problem more precisely we must define an ensemble of random Hamiltonians (or formulas). This is best done in the language of random factor graphs and is beneficial because it turns out that important typical properties are determined by typical geometric properties of these factor graphs.  
A factor graph is a bipartite graph with labeled ``variable nodes" $\{1, \dots, N\}$(associated to qubits or Boolean variables) and labeled ``constraint nodes"
$\{1, \dots, M\}$ (associated to projectors or disjunctions). For each constraint node $m$ we pick a $k$-tuple of variable nodes $(m_1, \dots , m_k)$ uniformly at random among all ${N}\choose{k}$ possible ones, and draw edges 
$(m, m_1),\dots, (m, m_k)$. This ensemble of random graphs is denoted $G^k_{N, M}$.
There is a further level of randomness. In the classical problem once a graph from $G^k_{N,M}$ is chosen, the variables in each disjunction are negated/non-negated with probability one-half (this information is usually encoded in the graph as a dashed/undashed edge). In the quantum problem, once a graph is chosen from $G^k_{N,M}$, one  samples the state $\vert\Phi^m\rangle$ uniformly at random in the (complex) Hilbert space of $k$ qubits (this time this information is not encoded in the graph structure). In practice $\vert\Phi^m\rangle$ is sampled by generating a $2^k$-dimensional complex Gaussian vector with i.i.d $\mathcal{C}\mathcal{N}(0, 1)$ components and then is normalized to make it a unit norm.\footnote{$\mathcal{C}\mathcal{N}(0, 1)$ means that real and imaginary parts are independent and distributed as $\mathcal{N}(0, 1/2)$.}

The parameter $\alpha= M/N$ is called the constraint density of the ensemble. There is a large literature on the phase diagram of random $k$-SAT in the thermodynamic limit $N, M\to +\infty$ with $M/N \to\alpha$ fixed. Structural and algorithmic phase transitions, as well as their interplay, are largely determined, although many questions remain unanswered (we refer to \cite{Montanari-Mezard, krzakala2007gibbs, SlyKsat} for more information). For random $k$-QSAT the current state of knowledge is more rudimentary and is summarized in figure \ref{fig:phase_diagramQSAT}.

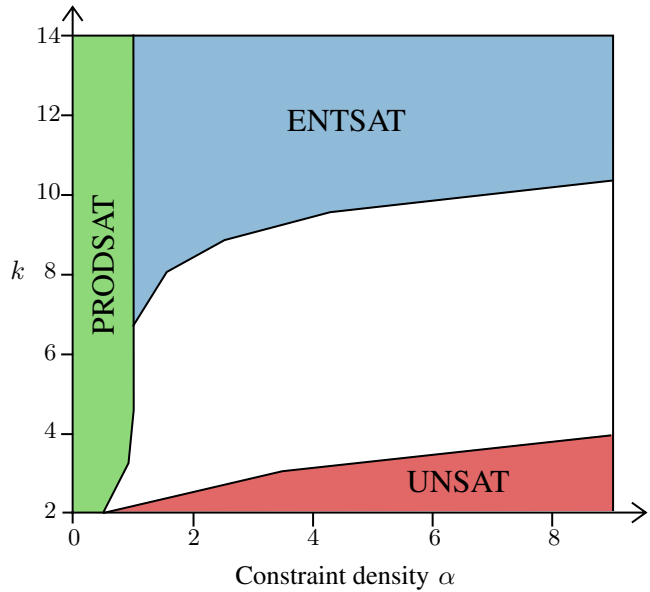
\begin{figure}
    \centering
    \tikzset{every picture/.style={line width=0.75pt}} 

\begin{tikzpicture}[x=0.75pt,y=0.75pt,yscale=-1,xscale=1]

\draw  [draw opacity=0][fill={rgb, 255:red, 214; green, 39; blue, 40 }  ,fill opacity=0.7 ] (165,248.75) -- (331.25,230.25) -- (331,268.75) -- (76.75,269) -- cycle ;
\draw  [draw opacity=0][fill={rgb, 255:red, 94; green, 200; blue, 64 }  ,fill opacity=0.7 ] (60.12,30.37) -- (90.26,29.76) -- (90.42,218.28) -- (87.81,244.93) -- (75,270) -- (60.12,269) -- cycle ;
\draw  [draw opacity=0][fill={rgb, 255:red, 31; green, 119; blue, 180 }  ,fill opacity=0.5 ] (331,30.37) -- (330.82,103.09) -- (189,119) -- (136,133) -- (107,149) -- (90.01,176.27) -- (90.26,29.76) -- cycle ;
\draw  (54.67,269.5) -- (347.44,269.5)(59.77,15.75) -- (59.77,274.98) (340.44,264.5) -- (347.44,269.5) -- (340.44,274.5) (54.77,22.75) -- (59.77,15.75) -- (64.77,22.75)  ;
\draw    (54,230.12) -- (60.12,230.12) ;
\draw    (54,190.12) -- (60.12,190.12) ;
\draw    (54,150.12) -- (60.12,150.12) ;
\draw    (54,110.12) -- (60.12,110.12) ;
\draw    (54,70.12) -- (60.12,70.12) ;
\draw    (54,30.12) -- (60.12,30.12) ;
\draw    (120.44,269.52) -- (120.44,276.52) ;
\draw    (180.44,269.52) -- (180.44,276.52) ;
\draw    (240.44,269.52) -- (240.44,276.52) ;
\draw    (300.44,269.52) -- (300.44,276.52) ;

\draw    (75,269.75) -- (87.81,244.68) -- (90.42,218.03) -- (90.26,29.51) ;
\draw    (75,269.75) -- (165,248.75) -- (330,230.75) ;
\draw    (90.01,176.02) -- (107,148.75) -- (136,132.75) -- (189,118.75) -- (330.82,102.84) ;
\draw   (60.12,30.12) -- (331,30.12) -- (331,268.75) ;

\draw (44,262.75) node [anchor=north west][inner sep=0.75pt]  [font=\scriptsize] [align=left] {$\displaystyle 2$};
\draw (44,222.75) node [anchor=north west][inner sep=0.75pt]  [font=\scriptsize] [align=left] {$\displaystyle 4$};
\draw (44,183.75) node [anchor=north west][inner sep=0.75pt]  [font=\scriptsize] [align=left] {$\displaystyle 6$};
\draw (44,142.75) node [anchor=north west][inner sep=0.75pt]  [font=\scriptsize] [align=left] {$\displaystyle 8$};
\draw (40,102.75) node [anchor=north west][inner sep=0.75pt]  [font=\scriptsize] [align=left] {$\displaystyle 10$};
\draw (40,63.75) node [anchor=north west][inner sep=0.75pt]  [font=\scriptsize] [align=left] {$\displaystyle 12$};
\draw (40,24.75) node [anchor=north west][inner sep=0.75pt]  [font=\scriptsize] [align=left] {$\displaystyle 14$};
\draw (56,276.75) node [anchor=north west][inner sep=0.75pt]  [font=\scriptsize] [align=left] {$\displaystyle 0$};
\draw (116,276.75) node [anchor=north west][inner sep=0.75pt]  [font=\scriptsize] [align=left] {$\displaystyle 2$};
\draw (176,276.75) node [anchor=north west][inner sep=0.75pt]  [font=\scriptsize] [align=left] {$\displaystyle 4$};
\draw (237,276.75) node [anchor=north west][inner sep=0.75pt]  [font=\scriptsize] [align=left] {$\displaystyle 6$};
\draw (297,276.75) node [anchor=north west][inner sep=0.75pt]  [font=\scriptsize] [align=left] {$\displaystyle 8$};
\draw (140,294.75) node [anchor=north west][inner sep=0.75pt]  [font=\footnotesize] [align=left] {{\fontfamily{ptm}\selectfont Constraint density} $\displaystyle \alpha $};
\draw (27,141.75) node [anchor=north west][inner sep=0.75pt]  [font=\footnotesize] [align=left] {$\displaystyle k$};
\draw (226,246.75) node [anchor=north west][inner sep=0.75pt]   [align=left] {{\fontfamily{ptm}\selectfont UNSAT}};
\draw (166,65.75) node [anchor=north west][inner sep=0.75pt]   [align=left] {{\fontfamily{ptm}\selectfont ENTSAT}};
\draw (68,181.75) node [anchor=north west][inner sep=0.75pt]  [rotate=-270] [align=left] {{\fontfamily{ptm}\selectfont PRODSAT}};

\end{tikzpicture}
    \caption{\small Phase diagram of QSAT as found in \cite{QSATShearerLaumann}. In the red region the problem is UNSAT and it is SAT in the green and blue regions. Furthermore it is PRODSAT in the green region and ENTSAT in the blue region (no satisfying product states but satisfying entangled states). The only line which is known to be exact is the one delimiting the green region: we have $\alphadc(k=3)\approx 0.92$ and for $k\geq 4$ this threshold monotonically approaches $1$ from below. The two lines delimiting the red and blue regions give respectively upper and lower bounds on the SAT-UNSAT threshold. The proof techniques of these bounds are adapted from corresponding proofs in the classical problem.}
    \label{fig:phase_diagramQSAT}
\end{figure}

The quantum phase diagram is richer than its classical counterpart already at the level of structural phase transitions, and almost nothing is known about algorithmic phase transitions. One may distinguish between PRODSAT, ENTSAT, and UNSAT phases. The UNSAT phase is simply the one where no zero energy eigenstates exist with high probability (w.h.p.). The SAT phase on the contrary is the one where zero energy eigenstates exist w.h.p. It can be decomposed into a PRODSAT phase for which there exist zero energy eigenstates which are fully factorized into a tensor product of single qubit states, and an ENTSAT phase where the zero energy eigenstates cannot be fully factorized into single qubit states (of course one could also envision more refined decompositions of the ENTSAT phase corresponding to partial ``non-single-qubit" factorizations). It is rigorously known that for $k=2$ the ENTSAT phase does not exist and that a sharp PRODSAT-UNSAT phase transition takes place with threshold $\alpha_{\rm c} =1/2$ \cite{QSATLaumann}. This is really a geometric transition, closely tied to the sudden proliferation of closed loops in the random factor graphs at this critical density. For $k\geq 3$ we only have loose upper and lower bounds for the various thresholds. In particular the existence of an ENTSAT phase is only proven for $k\geq 7$ and what happens for $3\leq k\leq 6$ is unclear. 

In a series of very interesting papers \cite{QSATLaumann, ProdSATLaumann},  it is shown that PRODSAT states exist if and only if the factor graph has a constraint-covering dimer configuration. A constraint-covering dimer configuration is a set of edges where all constraints are covered and no two edges meet at a constraint or at a variable node (but some variable nodes may not be covered). This is again a purely geometrical characterization of the PRODSAT phase. Using previous work on random graph theory \cite{KarpMatch, bordenave2013matchings}, this characterization allows to identify the maximum constraint density for which PRODSAT states exist w.h.p., as $ \alphaprod = \alphadc(k)$,  where $\alphadc(k)$ is the threshold corresponding to the existence of dimer coverings (in thermodynamic limit). However there remains the algorithmic question: can one efficiently find PRODSAT states and with what complexity? This question has been partly answered by looking at a purely geometric leaf-removal based algorithm which determines PRODSAT states in linear time $O(N)$ for $\alpha < \alphalr(k)< \alphadc(k)$. It has remained  unanswered for 
$\alphalr(k) < \alpha < \alphadc(k)$.

In this paper we concentrate on the PRODSAT phase and make the above picture fully rigorous.  

\begin{theorem}[Main Theorem]\label{mainthm}
Take a factor graph from the ensemble $G^k_{N,M}$. Given this graph take a set of $M$ projectors $\{\Pi_m, m=(m_1, \dots, m_k)\}$ uniformly at random. This defines a random formula or equivalently a Hamiltonian $H_F$ in \eqref{quantum_hamiltonian}. Let $\mathbb{P}_F$ be the probability with respect to this ensemble of random Hamiltonians. We have 
\begin{align*}
\lim_{N, M \to +\infty}\mathbb{P}_F\bigl(\exists \, \vert \Psi\rangle=\vert\varphi_1\rangle\otimes\dots\otimes \vert\varphi_N\rangle\in (\mathbb{C}^2)^{\otimes N}: H_F\vert \Psi\rangle=0\bigr) = 
\begin{cases} 
1, \,\,\, \alpha < \alphadc (k) \\
0, \, \,\, \alpha > \alphadc (k)
\end{cases}
\end{align*}
where the limit is such that $M/N \to \alpha$ fixed. 
\end{theorem}

\begin{remark}\label{dimerRemark}
    It is known \cite{KarpMatch, bordenave2013matchings} that the dimer covering threshold of the factor graph ensemble satisfies 
$\lim_{N, M \to +\infty}\mathbb{P}_{G^k_{N,M}}\bigl(\exists \emph{\, a dimer covering}\bigr) =  
1$ for $\alpha < \alphadc (k)$ and $0$ for $\alpha > \alphadc (k)$.
\end{remark}

The proof draws on ideas already present in \cite{QSATLaumann}. 
In section \ref{Main results} we first reformulate the problem on the $2$-core of the factor graph and explain the main strategy of proof.
For the direct statement we combine two arguments: one is purely analytical (section \ref{Analytical perturbative argument}) based on the implicit function theorem for functions of multiple complex variables, and the second is purely algebraic (section \ref{Algebraic non-perturbative argument}) based on Buchbergers's algorithm for solving algebraic polynomial equations. In the process we remark that the Buchberger algorithm can be used to provide w.h.p PRODSAT solutions for $\alphalr(k)<\alpha <\alphadc(k)$. As we will see randomness plays an important role in the analysis. The generic complexity of the algorithm is doubly exponential and it is an open problem to determine what is the real algorithmic complexity of finding these solutions in this range of densities. The proof of the converse statement for $\alpha > \alphadc(k)$ is presented in section \ref{Converse Statement}.

The nature of entanglement in this problem (beyond its mere existence for $k\geq 7$) is a largely open question. We make a few numerical observations for formulas with $N=M$ \textit{finite} and such that dimer coverings exist. These formulas have a finite number of PRODSAT states w.h.p., however our observations suggest that for a fraction of the formulas these states do not span the whole kernel space of the Hamiltonian. Therefore there exists a subspace of the kernel space which only contains entangled states. These observations are presented in section \ref{Simulations}.

For the remaining of this paper we say that an event $\mathcal{E}$ happens  w.h.p. if  $\lim_{N, M \to +\infty}\mathbb{P}\bigl(\mathcal{E}) = 1$ where the limit is such that $M/N\to\alpha$ and $\mathbb{P}$ is with respect to an ensemble that depends on the context. For example in Theorem \ref{mainthm} the ensemble corresponds to the random Hamiltonians (random factor graphs and projectors) whereas in remark \ref{dimerRemark} it simply corresponds to random factor graphs.

\section{Strategy of analysis and main results}\label{Main results}

We proceed with a two-stage process. First, given a graph from $G^k_{N,M}$ we simplify the constraint satisfaction problem to a problem where the numbers of constraints (or projectors) and qubits are equal. Second, this residual constraint satisfaction problem is reformulated as the study of solutions of a set of polynomial equations in complex variables.

\subsection{First step: $\alpha < \alphalr (k) $}

This step is accomplished using the leaf removal process, a Markov process in the space of factor graphs. Given an initial graph $G\in G^k_{N,M}$ one iteratively removes degree-one variable nodes together with its attached constraint node, until the residual graph has minimal variable-node degree at least two (the process then stops). This residual graph is called the core (equivalently, 2-core or hypercore). The theoretical analysis of this Markov process is well known and reviewed in Appendix \ref{Append_LR}. Lemma \ref{lemma_core} defines a threshold $\alphalr(k)$ such that for $\alpha<\alphalr (k)$ the core is empty w.h.p. and for $\alpha>\alphalr(k)$ the core is not empty w.h.p.. Let us denote by $G^\prime$ the residual graph and assume it contains $M^\prime\leq M$ constraint nodes and $N^\prime\leq N$ variable nodes. 

Given any {\it product} state $\ket{\Psi^\prime}\in (\mathbb{C}^2)^{\otimes N^\prime}$ for the qubits of the core (assuming it is non-empty), we construct a state 
\begin{align}
\ket{\Psi} =  \ket{\Psi^\prime}\otimes \prod_{i\in G\setminus G^\prime}\ket{\varphi_i}
\end{align}
where $\ket{\varphi_i}$ are single qubit states iteratively constructed thanks to Bravyi's transfer matrix by reversing the leaf removal steps. If $m$ is a previously deleted constraint node along with the deleted variable node $m_i$ where $(m_1,\dots, m_{i-1}, m_{i+1}, \dots, m_{k})$ are the remaining neighboring variable nodes, the transfer matrix is a linear map $T: (\mathbb{C}^2)^{\otimes k-1} \to \mathbb{C}^2$ such that 
\begin{align}
    \ket{\chi_{m_i}}:=T\bc{\ket{\chi_{m_1}}\otimes\dots \otimes \ket{\chi_{m_{i-1}}}\otimes \ket{\chi_{m_{i+1}}}\otimes\dots \otimes\ket{\chi_{m_{k}}}}
\end{align}
and the constraint \eqref{m_projector} is satisfied,
\begin{align}
    \langle\Phi^m\vert\chi_{m_1}\rangle\otimes \dots\otimes \ket{\chi_{m_{k}}} = 0
\end{align}
for any single qubit states $\ket{\chi_{m_1}},\dots, \ket{\chi_{m_{i-1}}}, \ket{\chi_{m_{i+1}}},\dots, \ket{\chi_{m_{k}}}$. That such a linear map $T$ exists and can be constructed explicitly is reviewed in Appendix \ref{Append_LR}.
Because we take a product state for $\vert\Psi^\prime\rangle$ we can apply this transfer matrix to the qubits involved in the last constraint removed and get the qubit state of the last variable node removed. We can then iterate this process following the reversed steps of leaf removal until all qubits are assigned a state thereby obtaining $\vert \Psi\rangle$.

When the core $G^\prime$ is empty, this construction yields a PRODSAT zero energy state simply by starting with an arbitrary tensor product of the $k-1$ qubits connected to the last removed leaf with its attached constraint. 

Thus, we have proved the intermediate result:

\begin{lemma}\label{lr-pre-lemma}
For $\alpha<\alphalr (k)$ there exist PRODSAT zero energy states w.h.p.. Furthermore their construction has time-complexity bounded by $O(N)$. 
\end{lemma}

\subsection{Second step: $ \alphalr (k) <\alpha < \alphadc (k) $}

From now on we assume the core $G^\prime$ is non-empty. In order to prove Theorem \ref{mainthm} we must solve a constraint satisfaction problem on this residual graph. 
More precisely we must show that there exists $\vert\Psi^\prime\rangle$ of {\it product} form such that for all $m\in G^\prime$
\begin{align}\label{constraint-on-core}
    \vert\Phi^m\rangle\langle \Phi^m\vert\otimes I_{N^\prime -k}\vert\Psi^\prime\rangle = 0
\end{align}
Let us relabel the variable and constraint nodes of $G^\prime$ as $\{1, \dots, N^\prime\}=[N^\prime]$ and $\{1, \dots, M^\prime\}=[M^\prime]$. Without loss of generality we can use the parameterizations
\begin{align}\label{tensor-product-state-param}
    \vert \Psi^\prime\rangle = \frac{\vert 0\rangle + z_1 \vert 1\rangle}{1+ \vert z_1\vert^2}\otimes \dots \otimes \frac{\vert 0\rangle + z_{N^\prime}\vert 1\rangle}{1+\vert z_{N^\prime}\vert^2} 
\end{align}
and \footnote{Here $\vert i_1\rangle_{m_1}= \vert 0\rangle_{m_1}, \vert 1\rangle_{m_1}$ are the computational basis states of the space $\mathbb{C}^2_{m_1}$ of qubit $m_1$, and similarly for the other qubits $m_2, \dots, m_k$ involved in $m$.} 
\begin{align}
    \vert \Phi^m\rangle = \sum_{(i_1,\dots, i_k)\in \{0, 1\}^k} \phi^m_{i_1\dots i_k} \vert i_1\rangle_{m_1}\otimes\dots\otimes \vert i_k\rangle_{m_k}
\end{align}
where $z_1, \dots, z_{N^\prime}$ and $\phi^m_{i_1\dots i_k}$ are complex numbers and 
$\sum_{(i_1,\dots, i_k)\in \{0, 1\}^k} \vert\phi^m_{i_1\dots i_k}\vert^2 = 1$. It is not difficult to see that the constraints \eqref{constraint-on-core} then become
\begin{align}\label{alg-equations}
    \sum_{(i_1,\dots, i_k)\in \{0, 1\}^k} \phi^m_{i_1\dots i_k} z_{m_1}^{i_1}\dots z_{m_k}^{i_k} = 0, \qquad m\in [M^\prime]
\end{align}
Note that the normalization constraint for the coefficients $\phi^m_{i_1\dots i_k}$ can be dropped as we can always multiply each equation by an arbitrary real positive number. Hence the problem is reduced to solving a set of polynomial equations in the ring $\mathbb{C}[z_1, \dots, z_{N^\prime}]$. 

An important idea introduced in \cite{ProdSATLaumann} is that the existence of solutions for the system \eqref{alg-equations} is controlled by the presence of constraint-covering dimer configurations of the factor graph (as defined in the introduction).

\begin{proposition}\label{direct-prop}
    Let $G^\prime$ have a constraint-covering dimer configuration. Then the system
    \eqref{alg-equations} has a solution for almost all choices of complex coefficients $\{\phi^m_{i_1\dots i_k}, m\in [M^\prime], (i_1\dots i_k)\in \{0,1\}^k\}$.
\end{proposition}

In \cite{ProdSATLaumann} the starting point is the observation that Proposition \ref{direct-prop} is easy to check if the projector $\vert\Phi^m\rangle \langle \Phi^m  \vert = \vert \varphi^{m_1}\rangle\langle  \varphi^{m_1} \vert \otimes\dots\otimes \vert \varphi^{m_k}\rangle \langle  \varphi^{m_k} \vert $ has product form, i.e., 
$\phi^m_{i_1\dots i_k} = \varphi^{m_1}_{i_1} \dots \varphi^{m_k}_{i_k}$
because it suffices then to solve 
$\varphi^{m_*}_{0} + z_{m_*}\varphi^{m_*}_{1} = 0$, where $m_*$ is the unique variable node belonging to the dimer that covers clause $m$. This equation has a solution for almost all 
$\varphi^{m_*}_{0}, \varphi^{m_*}_{1}\in \mathbb{C}$. Then this result is extended to non-product states through a perturbative argument combined with abstract algebraic geometry theorems. The proof is non-constructive. 

Here we will proceed differently with a constructive proof which is the object of section \ref{Analytical perturbative argument} and \ref{Algebraic non-perturbative argument}. We first note that when $\{\phi^m_{0\dots0} = 0, m\in [M^\prime]\}$ the polynomial equations have the trivial solution $z_1=\dots = z_{N^\prime} =0$. In section \ref{Analytical perturbative argument} we show through complex analysis arguments that one can construct (with probability one) a unique solution for $\{\phi^m_{0\dots0} \neq 0, m\in [M^\prime]\}\in U$ a small enough open neighborhood (depending on the instance) of the origin in $\mathbb{C}^{M^\prime}$. It is also possible to give an explicit series expansion formula for the solution. Then, in section \ref{Algebraic non-perturbative argument} we extend the existence of this solution (again with probability one) through an analysis of Buchberger's algorithm for solving polynomial equations in the ring $\mathbb{C}[z_1,\dots, z_{N^\prime}]$. 

In section \ref{Converse Statement} we show a converse statement.

\begin{proposition}\label{converse-prop}
    Let $G^\prime$ have no clause-covering dimer configuration. Then for almost all choices of complex coefficients $\{\phi^m_{i_1\dots i_k}, m\in [M^\prime], (i_1\dots i_k)\in \{0,1\}^k\}$ the system of equations 
    \eqref{alg-equations} has no solution.
\end{proposition}

Putting together Lemma \ref{lr-pre-lemma}, Propositions \ref{direct-prop} and \ref{converse-prop}  we obtain Theorem \ref{mainthm}.

\begin{proof}[Proof of Theorem \ref{mainthm}]
Take an instance of a factor graph $G\in G_{N,M}^k$. Note that the instance has a dimer covering if and only if the residual core $G^\prime$ also has a dimer covering. This is in fact so at each intermediate step of leaf removal, because constraint nodes that are not removed remain degree $k$, and thus their dimer is left untouched. As a result for $\alpha<\alphadc$ w.h.p. the core $G^\prime$ has a dimer covering, and Proposition \ref{direct-prop} implies that w.h.p. there exist tensor product states \ref{tensor-product-state-param} with zero energy. Conversely for $\alpha>\alphadc$ the graph $G^\prime$ w.h.p. has no dimer covering and Proposition \ref{converse-prop} implies that w.h.p. there are no such product states of zero energy.
\end{proof}
\section{Analytical perturbative argument}\label{Analytical perturbative argument}
In this section we prove that when all the constant terms $\{\phi^m_{0,\ldots,0}, m\in [M^\prime]\}$ are small enough, the system of equations \eqref{alg-equations} (restricted to the core $G^\prime$) has a solution. 
Recall that when $\phi^m_{0,\ldots,0}=0$ for all $m\in [M^\prime]$ this is obvious as $\vec z = \bc{z_1, \ldots,z_{N'}} = \bc{0,\ldots,0}$ is a trivial solution. 
We will show that this trivial solution can be extended to a solution when not all $\{\phi_{0,\ldots,0}^{m}, m\in[M']\}$ are $0$ but small enough:

\begin{proposition} \label{lemma_existsol} 
Suppose that $G^\prime$ admits a clause-covering dimer configuration.
Then there exists $\eps > 0$ such that if $\abs{\phi_{0,\ldots,0}^{m}}<\eps$ for all $m \in [M^\prime]$ then there exists  a solution $\bc{z^*_1, \ldots, z^*_{N'}}$  to the system \eqref{alg-equations} and a corresponding tensor product state 
$\vert\Psi^\prime\rangle$ satisfying \eqref{constraint-on-core}. Note that $\epsilon$ depends on $G^\prime$ and 
$\{\phi^m_{i_1\ldots i_k}, m\in [M^\prime], (i_i,\ldots, i_k)\in \{0, 1\}^k\setminus (0,\ldots,0)\}$.
\end{proposition}

Before proving Proposition \ref{lemma_existsol} we make a {\it reduction} of the system \eqref{alg-equations} to a {\it square} system. For $\vec{z}=\bc{z_1, \ldots, z_{N'}} \in \CC^{N'}$ and $m\in [M^\prime]$ let
\begin{align} \label{eq_f}
	f^{m}\bc{\vz}= \sum_{\substack{i_1 \ldots i_k\in\{0,1\}^k\\ \phi_{0\ldots 0}^m = 0}}  \ \phi^{m}_{i_1\ldots i_k} z_{m_1}^{i_1} \ldots z_{m_k}^{i_k}, 
\end{align}
These functions are simply the polynomials in \eqref{alg-equations} without the constant terms
and $\bc{0,\ldots,0}$ is a common zero. 
The existence of the dimer covering (assumed in Proposition \ref{lemma_existsol}) guarantees that there is one variable node $m_\star$ in the neighborhood of constraint $m$ that matches it (i.e, $(m_\star, m)$ is a dimer). 
We {\it reduce} the system \eqref{eq_f} of $M'$ equations and  $N'$ variables to  a {\it square} system of $M'$ equations and $M'$ variables by assigning $z=0$ for all the variable nodes that are not in the dimer covering. 
The polynomials of the reduced system will be denoted as $\fmred$.  Note that $\fmred$ might contain fewer than $k$ variables for some $m$ (this happens when $m$ has neighboring variables not covered by the chosen dimer covering). 
The following relabelling of complex variables turns out to be useful: given the labelling $\{1, \ldots, M'\}$ of constraint nodes we relabel the complex variables associated to nodes in the dimer covering by $\cbc{z_{1_\star}, \ldots, z_{M'_\star}}$.
Finally we define the multivariate complex map 
\begin{align}
	\Fsq: \CC^{M'} &\rightarrow \CC^{M'} \\
	\vz=(z_{1_\star}, \ldots, z_{M'_\star}) & \mapsto (\fsq^1(\vz), \ldots, \fsq^{M'}(\vz))
\end{align} 
where $\vz= (z_{1_\star}, \ldots, z_{M'_\star})$ is the set composed only of variables in the dimer covering.
The Jacobian matrix of $\Fsq$ is the $M'\times M'$ matrix
\begin{align}\label{Jacobian_rectangle}
		\JFsq:=\bc{\parder{\fmred}{z_{j_\star}}}_{m\in [M'],j\in [M']}.
\end{align} 

A crucial remark is the following: with probability one $\phi^m_{i_1\ldots i_k}\neq 0$, thus for any $m\in [M^\prime]$ there are monomials in $\fmred$ containing the variable $z_{m_\star}$; in particular it is guaranteed that a linear monomial of the form $\phi^m_{0\ldots 010\ldots 0} z_{m_\star}$ is always present.  
Thus, with probability one, all the diagonal elements of the Jacobian $\parder{\fmred}{z_{m_\star}}$  are polynomials with a constant term $\phi^m_{0\ldots 010\ldots 0}$ and in particular,
\begin{align}
\parder{\fmred}{z_{m_\star}}\Bigg\vert_{\vz=\vec{0}} =  \phi^m_{0\ldots 010\ldots 0} \neq 0.
\end{align}

\begin{lemma}\label{lemma_howjac_nonzero}
	$G^\prime$ has a dimer covering if and only if $\JFsq(\vec{0})$ is a full rank matrix with ${\det\JFsq(\vec{0})\neq 0}$.
\end{lemma}

\begin{proof} 
Consider the Jacobian matrix at $\vz=\vec{0}$. A row $m\in [M']$ always contains {\it at most} $k$ non-vanishing elements among  $\phi^m_{10\ldots0}, \phi^m_{010\ldots0}, \ldots, \phi^m_{0\ldots01}$ corresponding to the neighboring nodes belonging to some dimer covering. 

We first prove the direct statement. Suppose $G^\prime$ has a dimer covering. Then as remarked above, with probability one each row contains {\it at least} one non-vanishing element and especially one on the main diagonal.
We run Gaussian elimination on $\JFsq\bc{\vec{0}}$ with the elements on the main diagonal as the pivot to obtain the matrix in row echelon form. At each step of the algorithm we linearly combine rows, and the new terms we get on the diagonal can only be polynomial functions of the $\phi$'s. Since these polynomials are holomorphic multivariate functions (of the $\phi$'s) they have zero locus of measure zero \cite{gunning1965analytic}.
Eventually we get an upper triangular matrix with non-zero terms on the main diagonal.  
This proves that $\JFsq\bc{\vec{0}}$ is full row-rank, and since it is a square matrix $\det{\JFsq\bc{\vec{0}}}\neq 0$.

For the converse statement we must show that if $\JFsq\bc{\vec{0}}$ is full row-rank, then we can associate to each row $m\in [M']$ a column $j_m\in [M']$ such that the matrix element $(m, j_m)$ is non-zero and $j_m\neq j_{m'}$ for $m\neq m'$. The injective mapping $m\mapsto j_m$ provides the dimer covering.
Assume that no such injective mapping exists. Then, as we go down the rows, at a certain point, say for the $\widebar{m}$-th row, we cannot come up with $j_{\widebar m}$ such that $j_{\widebar m} \neq j_m$ for all $m<{\widebar m}$. This means all 
non-zero elements of the $\widebar m$-th row belong among the previously chosen columns $\{j_m | m<{\widebar m}\}$. Therefore the $\widebar m$-th row is a vector in the span of the previous rows $\{m < {\widebar m}\}$. This contradicts the full row-rank assumption.
\end{proof}

In the rest of this section we use results from multivariate complex analysis reviewed in Appendix \ref{appcomplex}. 

\begin{lemma}\label{cor:Jacnonzero}
Suppose $\det{J_{\Fsq}(\vec{0})} \neq 0$. Then there exist $\eps>0$ such that $\vec{0}$ is the only zero of the map $\Fsq$ in the open ball $B(\vec0,\eps)$. In other words $\vec{0}$ is an isolated zero of the map $F_{\rm sqr}$. 
\end{lemma}

\begin{proof}
By construction $\Fsq(\vec{0})=\vec 0$. Since each polynomial $\fmred$ is an holomorphic multivariate function we can use Theorem \ref{localinverse} to directly deduce the existence of $B(\vec0,\eps) \subset \CC^{M'} $  such that $\Fsq$ is biholomorphic in $B(\vec0,\eps)$. In particular $\Fsq\vert_{B(\vec0,\eps)}$ is a bijection from $B(\vec0,\eps)$ to $\Fsq(B(\vec0,\eps))$ so $\vec{0}$ is the only solution of $\Fsq(\vz)=\vec 0$ in $B(\vec0, \eps)$. This means $\vec{0}$ is an isolated zero. 
\end{proof}

We now turn to the proof of the main result of this section:

\begin{proof}[Proof of Proposition \ref{lemma_existsol}] 
Since the graph has a dimer covering, Lemma \ref{lemma_howjac_nonzero} implies that $\det{J_{\Fsq}(\vec{0})}$ is nonzero. So by Lemma \ref{cor:Jacnonzero}, $\vec{0}$ is an isolated zero in the open ball $B(\vec{0}, \epsilon)$. Choose $0<\epsilon' < \epsilon$ so that $\vec0$ is the only zero of $\Fsq$ in the closure of $B(\vec0,\epsilon')$.
Proposition \ref{prop:juzakov} then states that we can find 
$\varphi >0$ small enough such that the system of equations 
\begin{align}\label{eventual}
\fmred(\vz) + \phi^m_{0\ldots 0} = 0, \qquad m\in [M']
\end{align}
has simple zeros for almost all values of the constant terms in the set $\{\vert \phi^m_{0\ldots 0}\vert <\varphi, m\in [M']\}$. Moreover because $\vec{0}$ itself is a simple zero of the $F_{\rm sqr}$ (i.e., it is isolated and $\det{J_{\Fsq}(\vec{0})} \neq 0$) we deduce from Proposition \ref{Juz2} that the solution of equations \eqref{eventual} is unique for small enough constant terms. This implies the existence of a solution $\bc{z^*_1, \ldots, z^*_{N'}}$ for the full system \eqref{alg-equations} for small enough constant terms. The solution we have constructed here consists of $z^*_j=0$ if $j$ does not belongs to the dimer covering (recall the reduction step above) and $z^*_j$ the unique solution of \eqref{eventual} if $j$ belongs to the dimer covering. We note that while this solution for \eqref{eventual} is unique (for small enough constant terms) it is not unique for 
\eqref{alg-equations}. Indeed we could have done a similar construction by setting the $z_j$ variables of nodes $j$ not in the dimer covering to non-zero values.
\end{proof}
\section{Algebraic non-perturbative argument}\label{Algebraic non-perturbative argument}

In the previous section, we proved Lemma \ref{lemma_existsol} stating that if there exists a dimer covering of the interaction graph, then instances with small constant terms have a $\prodsat$solution w.h.p.. To extend this result to all possible instances, we will use Buchberger algorithm and Gröbner basis. These are powerful tools to solve systems of complex multivariate polynomial equations and hence also give a method to directly find the zeros of the system \eqref{alg-equations} of constraints. For the description of Buchberger algorithm, we refer to Appendix \ref{App:buchberger} and \cite{ideals}.  

\begin{definition}
    A polynomial is called generic if it is a polynomial of the form Eq. \ref{alg-equations} such that the coefficients of each monomial are taken uniformly at random on the unit sphere $\left(\CC^2\right)^{\otimes k}$.
\end{definition}

In Appendix \ref{App:buchberger}, we review the following corollary of Hilbert's Nullstellensatz.

\begin{corollary}\label{cor:null2}
	A set of polynomials in an algebraically closed field has no common zeros  if and only if the reduced Gröbner basis is $\{1\}$.
\end{corollary}

\begin{proposition} \label{prop:extension}
	Let $\mathcal{F}= (f_1,f_2,\dots,f_m ) $ be a set of generic polynomial equations in $K[X_0,\dots, X_n]$ that has a common solution, then for any given $a \in K$, $\mathcal{F}_a:=(f_1 + a,f_2,\dots,f_m)$ also have a common solution with probability 1 with respect to the distribution of the constituent coefficients of $\mathcal{F}$.
\end{proposition}
\begin{proof}
	If $\mathcal{F}=(f_1,\dots,f_m)$  have a common zero then by Corollary \ref{cor:null2}, there exists a Gröbner basis not reduced to 1 for $\mathcal{F}$. We want to show that $\mathcal{F}_a=(f_1 + a, f_2, \dots, f_m)$ will also have a Gröbner basis that is not reduced to 1, and therefore admitting a common solution.  This will be achieved by convincing ourselves that the Buchberger's algorithm applied on $\mathcal{F}$ and $\mathcal{F}_a$ produce the same steps in the sense that the monomials involved in each step for $\mathcal{F}$ and for $\mathcal{F}_a$ are identical with probability 1. 
	Let us analyze each step of Buchberger's algorithm. 
    
	\textit{Computation of the $S$-polynomial.} When we compute $S_{i,j}$ in $\mathcal{F}$ and $S^a_{i,j}$ in $\mathcal{F}_{a}$ (see step 6 in Algorithm \ref{alg:div}), the two lists of monomials in $S_{i,j}$ and $S^a_{i,j}$ will be the same with probability 1. Indeed, it could happen that the coefficients of the monomials in $S^a_{i,j}$ vanish but this puts algebraic constraint on the constituent coefficients of $\mathcal{F}$. With respect to the distribution of the constituent coefficients, the constraint is satisfied with probability 0.
	
	\textit{Reduction through multivariate division algorithm \ref{alg:div}.} 
     We should also check that in the steps  6 to 8 of Algorithm \ref{alg:div}, the monomials produced starting from $\mathcal{F}$ and those produced starting from $\mathcal{F}_a$ will be the same with probability 1. The only way that at some steps the monomials differ is that the coefficients of the monomials produced by $\mathcal{F}_a$ vanish. As previously, this puts an algebraic constraint on the constituent coefficients of $\mathcal{F}$ that would be satisfied with vanishing probability.

    We note that the number of steps in the algorithm is finite and therefore the number of algebraic constraints that stem from the application of Buchberger's algorithm to $\mathcal{F}_a$ is finite.  Thus such algebraic constraints make up a set of measure 0 of coefficients for $\mathcal{F}$.
\end{proof}

To illustrate this proof, we detail the steps of Buchberger's algorithm using an example of $2$-$\qsat$ on 3 variables.
\begin{align*}
	f_1 &= a_0 + a_1z_1 + a_2z_2 + a_{12}z_1z_2\\
	f_2 &= b_0 + b_2z_2 + b_3z_3 + b_{23}z_2z_3\\
	f_3 &= c_0 + c_1z_1 + b_3z_3 + c_{13}z_1z_3.
\end{align*}
\begin{example}[Computation of the $S$-polynomial]
	\begin{align*}
		LCM(f_1,f_2) &=  a_{12}b_{23}z_1z_2z_3\\
		S(f_1,f_2) &= b_{23}z_3 \cdot f_1 - a_{12}z_1 \cdot f_2\\
		& = - a_{12}b_0z_1  + a_0b_{23}z_3 - a_{12}b_2z_1z_2 + ( a_1b_{23}- a_{12}b_3 )z_1z_3 + a_2b_{23}z_2z_3 \\
		LCM(f_1 + a ,f_2) &=  a_{12}b_{23}z_1z_2z_3\\
		S(f_1+ a,f_2) &= b_{23}z_3 \cdot (f_1 + a)- a_{12}z_1 \cdot f_2\\
		& = - a_{12}b_0z_1  + (a_0 + a)b_{23}z_3 - a_{12}b_2z_1z_2 + ( a_1b_{23}- a_{12}b_3 )z_1z_3 + a_2b_{23}z_2z_3 
	\end{align*}
	In this example, we must avoid the event $a + a_0 = 0$ which would delete the monomial $z_3$. This event has probability $0$.
\end{example}

\begin{example}[Reduction through multivariate division algorithm]
	
	$S(f_1,f_2)$ (resp. $S(f_1 + a ,f_2)$) is successively reducible by $f_2,f_3$ and $f_1$ (resp. $f_2,f_3$ and $f_1+ a$). Set $A = (a_1b_{23}- a_{12}b_3)/c_{13}$. We have 
	\begin{align*}
		p_1 &= S(f_1,f_2)  - a_2\cdot f_2\\
		&= -a_2b_0 - a_{12}b_0z_1 - a_2b_2z_2 + (a_0b_{23}-a_2b_3)z_3 - a_{12}b_2z_1z_2 + ( a_1b_{23}- a_{12}b_3 )z_1z_3 \\
		p_2 &= S(f_1,f_2)  - a_2\cdot f_2 - \frac{a_1b_{23}- a_{12}b_3 }{c_{13}} \cdot f_3 = S(f_1,f_2)  - a_2\cdot f_2 - A \cdot f_3\\
		&= -(a_2b_0 +c_0A) - (a_{12}b_0+ c_1A )z_1 -a_2b_2z_2 + (a_0b_{23}-a_2b_3 - b_3A)z_3 - a_{12}b_2z_1z_2  \\
		p_3 &= S(f_1,f_2)  - a_2\cdot f_2 - A \cdot f_3 + b_2 \cdot f_1\\
		&= (a_0b_2 - a_2b_0 - c_0A ) + (a_1b_2 - a_{12}b_0 - c_1A )z_1 + (a_0b_{23}-a_2b_3 - b_3A)z_3 \\
	\end{align*}
	\begin{align*}
		p_1^a &= S(f_1 +a,f_2)  - a_2\cdot f_2\\
		&= -a_2b_0 - a_{12}b_0z_1 - a_2b_2z_2 + ((a_0+a)b_{23}-a_2b_3)z_3 - a_{12}b_2z_1z_2 + ( a_1b_{23}- a_{12}b_3 )z_1z_3 \\
		p_2^\alpha &= S(f_1+a,f_2)  - a_2\cdot f_2 - A \cdot f_3\\
		&= -(a_2b_0 +c_0A) - (a_{12}b_0+ c_1A )z_1 -a_2b_2z_2 + ((a_0+a)b_{23}-a_2b_3 - b_3A)z_3 - a_{12}b_2z_1z_2  \\
		p_3^a &= S(f_1 + a,f_2)  - a_2\cdot f_2 - A \cdot f_3 + b_2\cdot (f_1 +a)\\
		&= ((a_0+a) b_2 - a_2b_0 - c_0A ) + (a_1b_2 - a_{12}b_0 - c_1A )z_1 + ((a_0+a)b_{23}-a_2b_3- b_3A)z_3  
	\end{align*}
	The algebraic constraints in this example are  $a_2b_3 - (a_0 + a)b_{23} = 0 $ at step 1, $(a_2b_3 - b_3A) - (a_0 + a)b_{23} = 0$ at step 2 and 3 and $(a_2b_0 + c_0A) - (a_0 + a)b_2 = 0$ at step 3. These occur with probability $0$.
\end{example}

The proof of Proposition \ref{direct-prop} follows then directly from Propositions \ref{lemma_existsol}, \ref{prop:extension}. Thus we have proved that if there exists a dimer covering then there exists a $\prodsat$solution w.h.p.. 

\section{Converse Statement}\label{Converse Statement}

We prove Proposition \ref{converse-prop}. 
The proof relies on Hall's marriage theorem \cite{hall1935matchings} stated below and on the Macaulay resultant of a system of polynomials \cite{Macaulay1902SomeFI}. For a system of homogeneous polynomial equations of the same number of equations and variables with coefficients in an algebraically closed field (here $\mathbb{C})$, the Macaulay resultant is a polynomial of the coefficients which vanishes if and only if the system of equations has a common non-zero solution. For more details on the resultant and its property we refer to \cite[Chap 3. \S 2]{Cox2004algebraic}.

\begin{theorem}[Hall's marriage theorem]
    \label{th:hall}
	For a bipartite graph $(V, E) = (A \cup B, E)$, the following conditions are equivalent.
	\begin{itemize}
		\item There is a perfect matching of $A$ into $B$.
		\item For each $S \subseteq A$, the inequality $|S| \leq |N(S)|$ holds where $N(S)$ denotes the neighboring nodes of $S$ in $B$.
	\end{itemize}
\end{theorem}

\begin{remark}
    A perfect matching `of $A$ into $B$' is a dimer configuration which covers all nodes in $A$ (but not necessarily all nodes of $B$) such that no two edges have common nodes. 
\end{remark}

\begin{proof}[Proof of Proposition \ref{converse-prop}]
We apply Theorem \ref{th:hall} to the factor graph $G^\prime$ with $A=[M^\prime]$ the set of constraint nodes and $B = [N^\prime]$ the set of variable nodes. Thus there exists a constraint-covering dimer configuration if and only if for any subset $S\subseteq [M^\prime]$ the number of variables appearing in those constraints satisfies $\vert S\vert \leq \vert N(S) \vert  $. Taking the contrapositive, if $G^\prime$ has no  constraint-covering dimer configuration, there must exist a subset $S \subset [M^\prime]$ with $\vert N(S)  \vert<\vert S\vert$. One can find a subset $S' \subseteq S$ with $\vert S'\vert=\vert N(S) \vert +1$ constraints which contains all the  variables of $N(S)$. This set $ S'$ corresponds to a system of $\vert N(S)\vert +1$ polynomial equations of the form \ref{alg-equations} with $\vert N(S) \vert $ variables. Now we show that this overdetermined system of equations does not admit a solution which implies that the full system cannot admit a solution.

Take the polynomials corresponding to $S^\prime$, with variables relabeled as $z_1, \dots, z_{\vert N(S) \vert}$,
and make them homogeneous by introducing an additional variable $z_0$, as follows
\begin{align}
    z_0^k \sum_{(i_1,\dots, i_k)\in \{0, 1\}^k} \phi^m_{i_1\dots i_k} \bc{\frac{z_{m_1}}{z_0}}^{i_1}\dots \bc{\frac{z_{m_k}}{z_0}}^{i_k}, \qquad m\in S^\prime
\end{align}
Suppose now that the system of original equations has a common solution $(z_1^*, \dots, z_{\vert N(S) \vert }^*)$. Then the system of homogeneous equations also has a common solution $(z_0=1, z_1^*, \dots, z_{\vert N(S)\vert }^*)$ and this solution is not the zero solution (since $z_0=1$). Therefore the Macaulay resultant of the homogeneous system must vanish. However this resultant itself is a polynomial in the variables $\{\phi^m_{m_1, \dots, m_k}, m\in S^\prime\}$ and is an holomorphic function. The zero locus of an holomorphic function has measure $0$ \cite{gunning1965analytic} and therefore the Macaulay resultant does not vanish with probability $1$. Hence with probability $1$ the system of equations cannot have a common solution.
\end{proof}
\section{Simulations}\label{Simulations}

In this section, we investigate two issues in order to better understand the nature of the PRODSAT phase and its possible transition towards the ENTSAT phase. For $k\geq 8$ it is established that the ENTSAT phase exists but this is still open for lower $k$. The discussion in this section applies to any $k$ but we run simulations for $k=3$, as they are too costly in practice for higher values.

\subsection{Dimer coverings and dimension of solution space}

Theorem \ref{mainthm} establishes a precise connection between the presence of a dimer covering and the existence of a PRODSAT solution. It is therefore of interest to further investigate if the {\it structure} of the interaction graph can provide insights about the dimension of the solution space of $H_F$. 

We first gather a few observations about $\dimKerH$.
There are two sources of randomness in k-QSAT: the interaction graph and the choice of the projectors. For a fixed interaction graph, let us consider the corresponding Hamiltonian where the coefficients of the projectors are to be thought as \emph{indeterminates}. The $2^N\times 2^N$ Hamiltonian matrix $H_F$ (\ref{quantum_hamiltonian}) is sparse when it is represented in the computational basis since the projectors are $k$-local. Then the determinants of the $s \times s$ submatrices are polynomials in the indeterminates. Two situations may arise. These polynomials may be equal to a trivially `null polynomial' or to a bona fide non-trivial polynomial. 
Let $S$ be the largest $s$ such that there exists an $s \times s$ submatrix $M_{H_F}$ whose determinant is a non-trivial polynomial. Over the choices of the random projectors, the determinant of $M_{H_F}$ will vanish on a set of measure $0$. Thus the rank of $H_F$ will take the value $S$ with probability 1. Therefore, for any fixed interaction graph, we have $\dimKerH =2^N - S$ with probability 1 over the choices of the random projectors.\footnote{This is nothing other than the content of the geometrization theorem \cite{QSATLaumann}.} 
Note also that by the definition of $S$, for any given instance of the random projectors, ${\rm rank} H_F \leq S$ necessarily so 
that  $2^N - S \leq \dimKerH$.

We would like to compute the `generic' value of $\dimKerH$ which is $2^N-S$. This is not easy in general. Nevertheless, by the above remarks, it is certainly upper bounded by $\dimKerH$ for {\it separable projectors} (i.e., $\vert \Phi^m\rangle$ is a product state). This is interesting because it is easier to compute $\dimKerH$ for separable projectors, at least for a few simple graphs. It is not clear a priori when this upper bound is an equality because separable projectors form a set of measure zero in the space of all projectors, but numerical simulations suggest that this is so for the graphs reviewed. Figure \ref{fig:pattern} and Table \ref{tab:pattern} show a set of graphs and corresponding recurrence relations for $\dimKerH$ (denoted $r_m$) and also for the number of dimer coverings (denoted $d_m$). 

\begin{figure}
    \centering
    \input{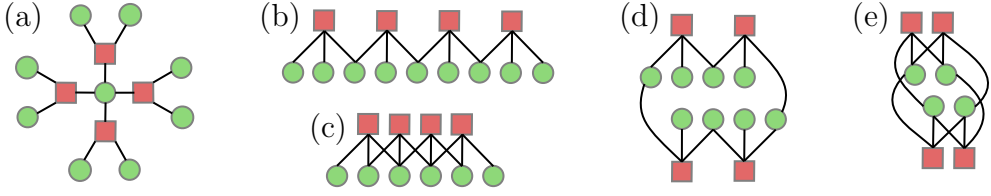}
    \caption{ Different patterns for $k= 3, m = 4$.
    (a) Sunflower (b) Loose chain (c) Strong chain (d) Loose cycle (e) Strong cycle. For the strong chain and the strong cycle, each qubit is connected to $k$ clauses, except for the boundary qubits. The sunflower is constructed around one central qubit.}
    \label{fig:pattern}
\end{figure}

\begin{table}
    \centering
    
{\scriptsize
\begin{tabular}{|l|cc|cc|}
\multicolumn{1}{l}{}& \multicolumn{2}{c}{Dimension of the kernel}         & \multicolumn{2}{c}{Dimer covering}                \\ \hline
Graph            & Initial values       & $r_m $             & Initial values             & $d_m $             \\ \hline
Sunflower$^\dag$ & $r_1 = 7$            & $3r_{m-1}+3^{m-1}$  & $d_1 = 3$                  & $2^{m-1}(m + 2)$     \\
Loose chain*     & $r_1 = 7, r_2 = 24$  & $4r_{m-1}-2r_{m-2}$ & $d_1 = 3, d_2 = 8$         & $3d_{m-1}-d_{m-2}$   \\
Loose cycle*     & $r_2 = 12, r_3 = 40$ & $4r_{m-1}-2r_{m-2}$ & $d_2 = 3, d_3 = 18 $       & $3d_{m-1}-d_{m-2}$   \\
Strong chain**    & $r_1 = 7, r_2 = 12$  & $r_{m-1}+r_{m-2}+1$ & $d_1 = 3, d_2 = 7, d_3=14$  & $2d_{m-1}-d_{m-3}+1$ \\
Strong cycle**   & $r_4 = 8, r_5 = 12$  & $r_{m-1}+r_{m-2}-1$ & $d_4 = 9, d_5 = 13, d_6=20$ & $2d_{m-1}-d_{m-3}$   \\ \hline
\end{tabular}\\
\medskip
}\label{tab:pattern}
    \caption{Recurrence relations for patterns in Fig. \ref{fig:pattern} with $k=3$ and $m$ the number of clauses.
    Regarding the dimension of the null space: $\dag$ the recurrence relation is proved in \cite{bounds} for all projectors (including non-separable ones); $*$ the two recurrence relations are proved in Appendix \ref{app:pattern} for separable projectors - numerical simulations give the dimension of the null space equal to this upper bound;
    $**$ are deduced by numerical simulation.}
\end{table}

These intriguing relations unfortunately do not seem to clearly demonstrate a general link between the number of dimer coverings and the dimension of the null space. Within this limited set of graphs, for a given graph type, we observe either $r_m \geq d_m$ or $r_m \leq d_m$ for all $m$. We have not found a universal relation between $r_m$ and $d_m$ beyond these inequalities. For example, we have $r_m = d_m - 1$ in the case of the strong cycle and $r_m = d_m+m+3$ in the case of the strong chain.

For $k=2$, the only satisfiable graphs are the tree and the cycle (it is easy to check that two intersecting cycles are not satisfiable \cite{QSATLaumann}). In that case, it is known that there is a gap between $\dimKerH = N+1 $ of a tree and $\dimKerH = 2$ of a cycle. However this linear growth of the gap does not seem to persist for $k=3$. Indeed, in solving the recurrence relations, we observe that the dimension of the null space for all patterns grows exponentially in $m$ (this exponential growth just follows from the order $2$ relations).
\subsection{PRODSAT basis}\label{PRODSATbasis}
We now wish to discuss \emph{how much entanglement is present} in the PRODSAT phase by comparing the dimension of the space generated by the product solutions with that of the full solution space.

A k-QSAT instance is PRODSAT if it is satisfied by a product state. However, this does not imply that all the solutions to the problem are product states. Indeed, the (normalized) sum of two different product states is still a solution to the problem and is likely to be entangled. For a given instance of random $k$-QSAT, the space generated by all the PRODSAT solutions is referred to as the {\it PRODSAT space}. A basis of the full solution space, $\ker H_F$, is said to be a {\it fully PRODSAT basis}, if all the vectors of the basis are product states. Let $\dimProd$ denote the dimension of the PRODSAT space.

An interesting question is the following: 
Is it true that the kernel space admits a fully PRODSAT basis?
While we do not directly study an ENTSAT phase in this paper, this question is clearly motivated by the harder issue of how a PRODSAT phase potentially disappears in favor of an ENTSAT phase.

Here we address this question in the following restricted setting of finite sizes with $N=M$ and $k = 3$.  Note that although $M/N = 1$, we are dealing here with finite size, so there exist instances with dimer coverings which are therefore PRODSAT. In particular, for $M=N=5,6$ it is known that {\it all} graphs have dimer coverings.

Even in the restricted setting $N=M$ and $k = 3$, it is not easy to compare $\dimProd$ and $\dimKerH$, and here this is done only for moderate sizes up to $N=M=11$. Indeed, to obtain $\dimKerH$ we use exact diagonalization to count the zero eigenvalues of the Hamiltonian which costs roughly $O(2^{3N})$ operations. At the same time, the computation of the PRODSAT solutions can be achieved through Buchberger's algorithm which requires a substantial runtime even for moderate sizes. Instead, we will rely on the BKK theorem (\ref{thm:bkk}) to obtain only the number of PRODSAT solutions of Eq. \ref{alg-equations}. Before stating the theorem, we need to recall the following: 

\begin{definition}
    The Newton polytope of a polynomial $f = \sum_{\alpha \in \Gamma} c_\alpha x^\alpha, \Gamma \subset \ZZ^n$ is the polytope formed by the convex hull of the set of all $\alpha \in \Gamma$. For polytopes $P_1, \dots, P_n$, the Mixed Volume $MV_n(P_1, \dots, P_n)$ is the coefficient of the monomial $\lambda_1 \dots\lambda_n$ in the polynomial $f(\lambda_1,\dots, \lambda_n) = Vol_n(\lambda_1P_1+ \dots + \lambda_n P_n)$ where the $+$ represents the Minkowski sum. Figure \ref{fig:mixed_volume} is an example of these definitions. For a k-QSAT instance, we denote by $MV$ the mixed volume of the polytopes associated with the polynomials in Eq. \ref{alg-equations}.
\end{definition}
\begin{figure}[H]
    \centering
    \tikzset{every picture/.style={line width=0.75pt}} 

\begin{tikzpicture}[x=0.75pt,y=0.75pt,yscale=-1,xscale=1]

\draw  [draw opacity=0][fill={rgb, 255:red, 74; green, 144; blue, 226 }  ,fill opacity=0.8 ] (328,132) -- (378,132) -- (378,182) -- (328,182) -- cycle ;
\draw  [draw opacity=0][fill={rgb, 255:red, 80; green, 64; blue, 209 }  ,fill opacity=0.77 ] (378,72.5) -- (407.5,132) -- (378,132) -- cycle ;
\draw  [draw opacity=0][fill={rgb, 255:red, 94; green, 200; blue, 64 }  ,fill opacity=0.7 ] (328,72.75) -- (378,72.75) -- (378,132) -- (328,132) -- cycle ;
\draw  [draw opacity=0][fill={rgb, 255:red, 94; green, 200; blue, 64 }  ,fill opacity=0.7 ] (378,132) -- (407.5,132) -- (407.5,182) -- (378,182) -- cycle ;
\draw   (328.5,182.08) .. controls (328.5,186.75) and (330.83,189.08) .. (335.5,189.08) -- (342.94,189.08) .. controls (349.61,189.08) and (352.94,191.41) .. (352.94,196.08) .. controls (352.94,191.41) and (356.27,189.08) .. (362.94,189.08)(359.94,189.08) -- (371,189.08) .. controls (375.67,189.08) and (378,186.75) .. (378,182.08) ;
\draw   (378,182.08) .. controls (377.94,186.06) and (379.9,188.08) .. (383.89,188.14) -- (383.89,188.14) .. controls (389.57,188.22) and (392.38,190.25) .. (392.33,194.23) .. controls (392.38,190.25) and (395.25,188.3) .. (400.94,188.38)(398.39,188.35) -- (400.94,188.38) .. controls (404.92,188.44) and (406.94,186.48) .. (407,182.5) ;
\draw   (328.25,132.08) .. controls (323.58,132.08) and (321.25,134.41) .. (321.25,139.08) -- (321.25,146.52) .. controls (321.25,153.19) and (318.92,156.52) .. (314.25,156.52) .. controls (318.92,156.52) and (321.25,159.85) .. (321.25,166.52)(321.25,163.52) -- (321.25,174.58) .. controls (321.25,179.25) and (323.58,181.58) .. (328.25,181.58) ;
\draw   (328,73.33) .. controls (323.33,73.33) and (321,75.66) .. (321,80.33) -- (321,92.21) .. controls (321,98.88) and (318.67,102.21) .. (314,102.21) .. controls (318.67,102.21) and (321,105.54) .. (321,112.21)(321,109.21) -- (321,125.08) .. controls (321,129.75) and (323.33,132.08) .. (328,132.08) ;
\draw  [draw opacity=0][fill={rgb, 255:red, 74; green, 144; blue, 226 }  ,fill opacity=0.8 ] (50,110.25) -- (90.21,110.25) -- (90.21,150.46) -- (50,150.46) -- cycle ;
\draw  [draw opacity=0][fill={rgb, 255:red, 80; green, 64; blue, 209 }  ,fill opacity=0.77 ] (162.21,90.25) -- (202.21,170) -- (162.21,170) -- cycle ;
\draw  [fill={rgb, 255:red, 0; green, 0; blue, 0 }  ,fill opacity=1 ] (49,110.63) .. controls (49,109.73) and (49.73,109) .. (50.63,109) .. controls (51.52,109) and (52.25,109.73) .. (52.25,110.63) .. controls (52.25,111.52) and (51.52,112.25) .. (50.63,112.25) .. controls (49.73,112.25) and (49,111.52) .. (49,110.63) -- cycle ;
\draw  [fill={rgb, 255:red, 0; green, 0; blue, 0 }  ,fill opacity=1 ] (88.58,110.25) .. controls (88.58,109.35) and (89.31,108.63) .. (90.21,108.63) .. controls (91.1,108.63) and (91.83,109.35) .. (91.83,110.25) .. controls (91.83,111.15) and (91.1,111.88) .. (90.21,111.88) .. controls (89.31,111.88) and (88.58,111.15) .. (88.58,110.25) -- cycle ;
\draw  [fill={rgb, 255:red, 0; green, 0; blue, 0 }  ,fill opacity=1 ] (49,150.63) .. controls (49,149.73) and (49.73,149) .. (50.63,149) .. controls (51.52,149) and (52.25,149.73) .. (52.25,150.63) .. controls (52.25,151.52) and (51.52,152.25) .. (50.63,152.25) .. controls (49.73,152.25) and (49,151.52) .. (49,150.63) -- cycle ;
\draw  [fill={rgb, 255:red, 0; green, 0; blue, 0 }  ,fill opacity=1 ] (88,150.63) .. controls (88,149.73) and (88.73,149) .. (89.63,149) .. controls (90.52,149) and (91.25,149.73) .. (91.25,150.63) .. controls (91.25,151.52) and (90.52,152.25) .. (89.63,152.25) .. controls (88.73,152.25) and (88,151.52) .. (88,150.63) -- cycle ;
\draw  [fill={rgb, 255:red, 0; green, 0; blue, 0 }  ,fill opacity=1 ] (161,169.63) .. controls (161,168.73) and (161.73,168) .. (162.63,168) .. controls (163.52,168) and (164.25,168.73) .. (164.25,169.63) .. controls (164.25,170.52) and (163.52,171.25) .. (162.63,171.25) .. controls (161.73,171.25) and (161,170.52) .. (161,169.63) -- cycle ;
\draw  [fill={rgb, 255:red, 0; green, 0; blue, 0 }  ,fill opacity=1 ] (200,169.63) .. controls (200,168.73) and (200.73,168) .. (201.63,168) .. controls (202.52,168) and (203.25,168.73) .. (203.25,169.63) .. controls (203.25,170.52) and (202.52,171.25) .. (201.63,171.25) .. controls (200.73,171.25) and (200,170.52) .. (200,169.63) -- cycle ;
\draw  [fill={rgb, 255:red, 0; green, 0; blue, 0 }  ,fill opacity=1 ] (160.58,90.25) .. controls (160.58,89.35) and (161.31,88.63) .. (162.21,88.63) .. controls (163.1,88.63) and (163.83,89.35) .. (163.83,90.25) .. controls (163.83,91.15) and (163.1,91.88) .. (162.21,91.88) .. controls (161.31,91.88) and (160.58,91.15) .. (160.58,90.25) -- cycle ;
\draw  [fill={rgb, 255:red, 0; green, 0; blue, 0 }  ,fill opacity=1 ] (160.75,130.13) .. controls (160.75,129.23) and (161.48,128.5) .. (162.38,128.5) .. controls (163.27,128.5) and (164,129.23) .. (164,130.13) .. controls (164,131.02) and (163.27,131.75) .. (162.38,131.75) .. controls (161.48,131.75) and (160.75,131.02) .. (160.75,130.13) -- cycle ;

\draw (384.5,115.25) node [anchor=north west][inner sep=0.75pt]  [font=\scriptsize] [align=left] {$\displaystyle \lambda _{2}^{\ 2}$};
\draw (338,99.25) node [anchor=north west][inner sep=0.75pt]  [font=\scriptsize] [align=left] {$\displaystyle 2\lambda _{1} \lambda _{2}$};
\draw (379.5,150.25) node [anchor=north west][inner sep=0.75pt]  [font=\scriptsize] [align=left] {$\displaystyle \lambda _{1} \lambda _{2}$};
\draw (348.25,196.63) node [anchor=north west][inner sep=0.75pt]  [font=\scriptsize] [align=left] {$\displaystyle \lambda _{1}$};
\draw (388.5,196.63) node [anchor=north west][inner sep=0.75pt]  [font=\scriptsize] [align=left] {$\displaystyle \lambda _{2}$};
\draw (346,150.25) node [anchor=north west][inner sep=0.75pt]  [font=\scriptsize] [align=left] {$\displaystyle \lambda _{1}^{\ 2}$};
\draw (297.25,151.63) node [anchor=north west][inner sep=0.75pt]  [font=\scriptsize] [align=left] {$\displaystyle \lambda _{1}$};
\draw (295.5,97.63) node [anchor=north west][inner sep=0.75pt]  [font=\scriptsize] [align=left] {$\displaystyle 2\lambda _{2}$};
\draw (28.75,152) node [anchor=north west][inner sep=0.75pt]  [font=\tiny] [align=left] {$\displaystyle ( 0,0)$};
\draw (69.75,153.25) node [anchor=north west][inner sep=0.75pt]  [font=\tiny] [align=left] {$\displaystyle ( 1,0)$};
\draw (27.25,112.75) node [anchor=north west][inner sep=0.75pt]  [font=\tiny] [align=left] {$\displaystyle ( 0,1)$};
\draw (69.5,114) node [anchor=north west][inner sep=0.75pt]  [font=\tiny] [align=left] {$\displaystyle ( 1,1)$};
\draw (140.75,172.5) node [anchor=north west][inner sep=0.75pt]  [font=\tiny] [align=left] {$\displaystyle ( 0,0)$};
\draw (180,172.25) node [anchor=north west][inner sep=0.75pt]  [font=\tiny] [align=left] {$\displaystyle ( 1,0)$};
\draw (140.5,92) node [anchor=north west][inner sep=0.75pt]  [font=\tiny] [align=left] {$\displaystyle ( 0,2)$};
\draw (92.5,96.75) node [anchor=north west][inner sep=0.75pt]  [font=\scriptsize] [align=left] {$\displaystyle A$};
\draw (92.5,137.75) node [anchor=north west][inner sep=0.75pt]  [font=\scriptsize] [align=left] {$\displaystyle B$};
\draw (53.5,137.75) node [anchor=north west][inner sep=0.75pt]  [font=\scriptsize] [align=left] {$\displaystyle D$};
\draw (53,96.75) node [anchor=north west][inner sep=0.75pt]  [font=\scriptsize] [align=left] {$\displaystyle C$};
\draw (165,76.25) node [anchor=north west][inner sep=0.75pt]  [font=\scriptsize] [align=left] {$\displaystyle E$};
\draw (203.25,156.75) node [anchor=north west][inner sep=0.75pt]  [font=\scriptsize] [align=left] {$\displaystyle F$};
\draw (164.5,155.5) node [anchor=north west][inner sep=0.75pt]  [font=\scriptsize] [align=left] {$\displaystyle G$};
\draw (61.5,45.08) node [anchor=north west][inner sep=0.75pt]   [align=left] {$\displaystyle P_{1}$};
\draw (172.25,45.08) node [anchor=north west][inner sep=0.75pt]   [align=left] {$\displaystyle P_{2}$};
\draw (313.83,45.08) node [anchor=north west][inner sep=0.75pt]   [align=left] {$\displaystyle \lambda _{1} P_{1} +\lambda _{2} P_{2}$};

\end{tikzpicture}
    \caption{Example of Mixed Volume Computation. The Newton polytope $P_1$ of the polynomial $f_1 = A xy + B x + Cy + D$ is the square with vertices $\{(1,1), (1,0), (0,1), (0,0)\}$. For $f_2 = Ey^2 + Fx + G$, it is a triangle with vertices $\{(0,2), (1,0), (0,0)\}$. The decomposition of $Vol_n(\lambda_1 P_1  +\lambda_2 P_2) = \lambda_1^2 + 3 \lambda_1\lambda_2 + \lambda_2^2$ is represented on the figure. Then the mixed volume is 3. }
    \label{fig:mixed_volume}
\end{figure}
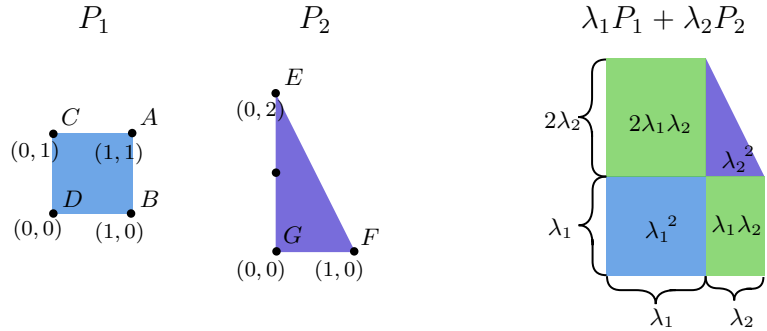

\begin{theorem}[BKK theorem \cite{Bernshtein1975TheNO}] \label{thm:bkk}
    Let $f_1,\dots, f_n$ be Laurent polynomials over $\CC$,  
    \begin{equation}
        f_i = \sum_{\alpha \in \Gamma_i} c_\alpha x^\alpha \qquad c_\alpha \in \CC, \quad \Gamma_i \subset \ZZ^k
    \end{equation}
    with finitely many common zeroes in $(\CC^\ast)^n$.
    Let $P_i$ be the Newton polytope of $f_i$. Then the number of common zeroes of the $f_i$ in $(\CC^\ast)^n$ is upper bounded by the mixed volume $MV_n(P_1, \dots, P_n)$. For generic choices of coefficients in $f_i$'s, the number of common solutions equals $MV_n(P_1, \dots, P_n)$.
\end{theorem}

\begin{figure}[H]
\centering
\begin{minipage}{.49\textwidth}
  \centering
  \includegraphics[trim={0 0 {4em} 0},clip,width=\linewidth]{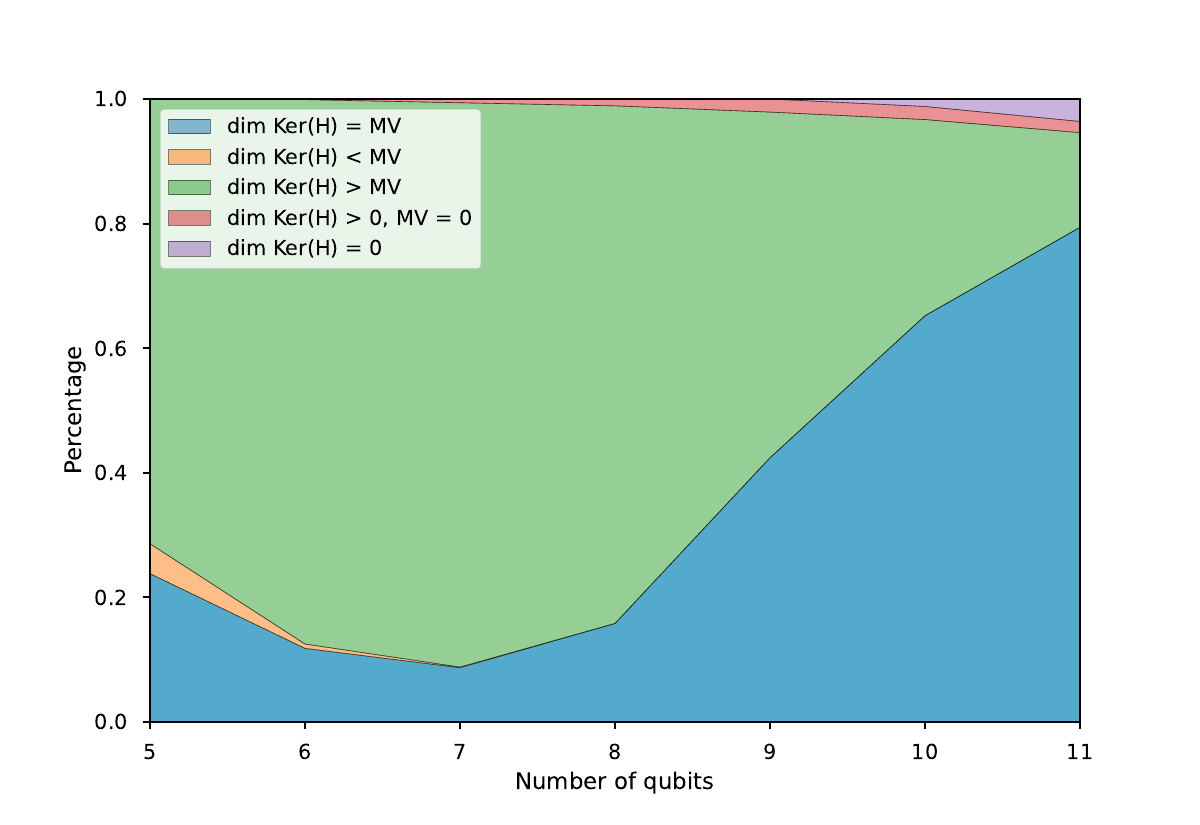}
\end{minipage}
\begin{minipage}{.49\textwidth}
  \centering
  \includegraphics[trim={0 0 {4em} 0},clip,width=\linewidth]{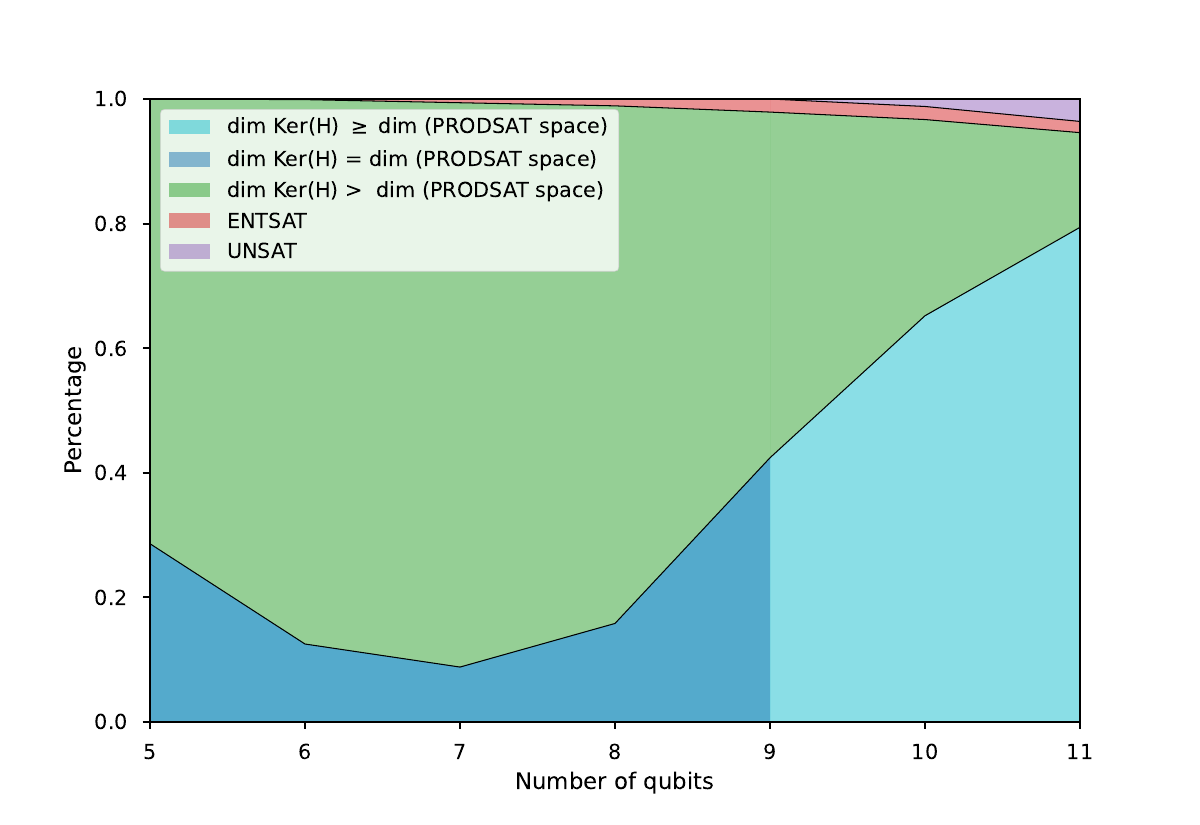}
\end{minipage}
\caption{Comparison among $\dimKerH, \dimProd$ and $MV$. 
On the left, $\dimKerH$ and  $MV$ are compared. On the right, we reinterpret the results in terms of PRODSAT space, UNSAT, and ENTSAT instances.
For blue and green, we use $\dimProd\leq MV$. Red and purple follow from the definition of the phases. Orange results are joined with blue since we always have $\dimKerH \geq \dimProd$.
For $N=5,6$ all the instances with $M=N$ are computed (resp. 252 and 38500). For $7\leq N \leq 10$, 5000 instances are sampled uniformly. For $N=11$, only 500 instances are used.}
\label{fig:prodsat_basis}
\end{figure}

\begin{remark}
For a given set of equations of the form \ref{alg-equations}, the corresponding mixed volume does not depend on the coefficients of projectors, but only on the monomials. Hence, for $k$-QSAT, the mixed volume only depends on the interaction graph. 
Regarding the complexity of computing the mixed volumes of polytopes, it is at least \#P-hard \cite{MV1998}.
\end{remark}

Since the product states obtained by substituting the $\{z_i\}$ solutions of Eq. \ref{alg-equations} in Eq. \ref{tensor-product-state-param} could be linearly dependent, the mixed volume only gives an upper bound on $\dimProd$,
\begin{equation}
\dimProd\leq MV.
\end{equation}
When we can compute the $\{z_i\}$ with Buchberger's algorithm, we can check whether this inequality is tight or not by looking for linear dependencies between PRODSAT solutions. Three situations can arise: 
\begin{itemize}
    \item $MV < \dimKerH$. Then $\dimProd \leq MV < \dimKerH$ so the basis is not fully PRODSAT. \item $MV = \dimKerH$. Then  $\dimProd \leq MV = \dimKerH$ so we cannot conclude if the basis is fully PRODSAT or not. 
    \item $MV > \dimKerH$. Then $\dimProd\leq \dimKerH < MV$ so there must be linear dependencies among PRODSAT solutions. We cannot conclude if the basis is fully PRODSAT or not.
\end{itemize}

Figure \ref{fig:prodsat_basis} shows the percentage of instances for which these scenarios occur for $M=N$ between $5$ and $11$. For increasing $N$, we observe an increase in the proportion of instances with $MV = \dimKerH$, which is somewhat unexpected (blue region). In particular, up to $N\leq 9$, we can check that $\dimProd=\dimKerH$ so the basis is indeed fully PRODSAT. Unfortunately, it is difficult to assess if this is still true for $N=10, 11$ but the trend in the figure suggests this might be so. This finding may seem rather surprising as one might have expected that the trend of the share of the green region increasing, observed for $N=5, 6, 7$, would continue with fully PRODSAT basis becoming rarer. 
We also find that for $N\geq 7$, there appear a small fraction of instances for which $\dimKerH>0$ and $MV=0$.  This corresponds to the existence of ENTSAT instances. For $N\geq 9$, there also appear a fraction of UNSAT instances. These results may point towards a picture of coexisting fractions of fully PRODSAT and non-fully PRODSAT instances in the large size limit $N, M\to +\infty$, $M/N=1$ for a random ensemble of instances {\it conditioned} on the existence of dimer coverings. 
  
\section{Conclusion}\label{Conclusion}

In this work, we have provided a comprehensive analysis of the PRODSAT phase with zero-energy eigenstates of product-form which are present for clause densities $\alpha < \alpha_{dc}(k)$, and disappear for $\alpha > \alpha_{dc}(k)$. This is a geometric transition driven by the existence versus non-existence of dimer coverings for the graph ensemble. When the clause density exceeds $\alpha_{dc}(k)$, the emergence of zero-energy solutions, if they exist, necessarily implies the presence of an entangled satisfiable (ENTSAT) phase. While such a phase is known to exist for $k \geq 7$, and is absent  for $k=2$, the situation for $3 \leq k \leq 6$ remains unresolved from the mathematical standpoint. Note that the numerical simulations of \cite{ClustQSATLaumann} points toward a direct transition from PRODSAT to UNSAT for $k=3$.  

This work points to several open directions for further research. 

It is unclear whether the transition between the ENTSAT and UNSAT phases would also be driven by the geometric properties of the underlying graph ensemble, and if this would be true for all values of $k$ displaying such a transition. Establishing a connection between graph-theoretic characteristics and the entanglement structure of zero-energy states could help to understand the emergence of the ENTSAT phase. A related question is the one investigated in section \ref{PRODSATbasis}, namely does the ENTSAT phase somehow already emerge within the PRODSAT one. The explicit construction of entangled satisfying assignments (beyond numerical Hamiltonian diagonalization) presents a major challenge, and developing systematic methods or algorithms for their generation would significantly enhance our ability to probe the ENTSAT regime. 

It is easy to see that for $k$-QSAT instances whose projectors have product-form, when $\alpha < \alpha_{dc}$ there exists a zero-energy eigenstate. Numerical observations seem to indicate that more is true, namely if the degree of variable nodes is less than $k$, the dimension of the zero-energy eigenspace is the same whether we sample only within product-form projectors or general projectors. Establishing this fact rigorously for graph ensembles with fixed (appropriate) variable node degree is an open problem.

\FloatBarrier

\begin{appendix}
    \section{Leaf Removal Process ($\lr$)}\label{Append_LR}
In this appendix, we give the proof of Lemma \ref{lr-pre-lemma} for completeness. The proof is algorithmic and based on two ingredients, namely a leaf removal process on factor graphs and Bravyi's transfer matrix. 

We begin with a brief review of the leaf removal (LR) process on a factor graph. We first delete isolated (degree-zero) vertices.  Next, we choose a unary (degree-one) variable $v_1$ at random and delete it along with its sole neighbor $a\in \partial v_1$.  By doing so, we also remove the other $k-1$ edges of $a$ connected to $v_2,\dots ,v_k\in\partial a$.  Such removal possibly makes some or all of $v_2,\dots, v_k$ isolated or unary variables.  Delete the isolated variables once again and then start again at an unary variable to delete further on as described before.  We iterate this process until we cannot find any more isolated or unary variables.  The process concludes with a subgraph where each variable is connected to at least two checks while all the checks are still connected to $k$ variables.  This (possibly empty) subgraph is referred to as the $2$-core of the hypergraph $G$ (equivalently, core or hypercore). 

The study of the $2$-core was done in \cite{Molloy2005} where the existence of a threshold  $\alpha_{ \rm lr}(k)$ is proven, below which the 2-core is empty and above which it is not empty w.h.p.. More precisely, let $G_{N,	p= \alpha/N^{k-1}}^k$ be the underlying $k$-uniform hypergraph where each $ \binom{n}{k} $ possible edges appear with probability $p$,  we have the following lemma:
 
\begin{lemma}\cite[Theorem 1]{Molloy2005} \label{lemma_core} 
Define 
	$$\alphalr(k)=\min_{x>0}\frac{(k-1)!x }{ \bc{1-\eul^{-x}}^{k-1} }.$$
	\begin{enumerate}
		\item 	For any $\alpha<\alphalr$,  $G_{N,	p= \alpha/N^{k-1}}^k$ has no non-empty $2$-core \whp 
		\item For $\alpha>\alphalr$, $G_{N,	p= \alpha/N^{k-1}}^k$ has a $2$-core of size $\beta(\alpha) N + o(N)$ \whp, with
		 $ \beta(\alpha) = 1-\eul^{-x}-\eul^{-x} x, $
		where $x$ is the greatest solution of
		$$\alpha=\frac{(k-1)!x }{ \bc{1-\eul^{-x}}^{k-1}}.$$
	\end{enumerate}
\end{lemma}	

We note that the construction of $G_{N,p= \alpha/N^{k-1}}^k$ is a bit different from $G_{N,M}^k$ but the two random hypergraph models are mutually contiguous meaning that any events that happen \whp \ in $G_{N,p= \alpha/N^{k-1}}^k$ also happen \whp \ in $G_{N,M}^k$ and vice versa. 

The second ingredient needed is Bravyi's transfer matrix \cite{Bravyi}. As described above, we remove certain vertices and edges according to LR.  An inherent reason for removing them is that the removed constraints should be easily satisfied by the removed variables.  We show how to implement this idea here (see \cite{Bravyi} for $k=2$).  

\begin{lemma}\cite[for $k=2$]{Bravyi}\label{lemma_T}
	 For all projectors $\vert \Phi^m\rangle\langle \Phi^m\vert$ and any selected variable node $m_i$ involved in constraint $m\equiv \{m_1,\dots, m_k\}$, we can construct a transfer matrix $T$ of size $2 \times 2^{k-1}$ such for that given any product state $\chiket{m_1} \otimes \dots \otimes \chiket{m_{i-1}} \otimes \chiket{m_{i+1}}\dots \otimes \chiket{m_k}$, the constraint is satisfied
  $$
  \langle \Phi^m  \chiket{m_1} \otimes \dots \otimes \chiket{m_k} =0
  $$
  for 
  $$
  \chiket{m_i} \propto T \chiket{m_1}\otimes \dots \otimes \chiket{m_{i-1}}\otimes\chiket{m_{i+1}}\otimes\dots \chiket{m_k}
  $$
  The proportionality sign indicates that the state still has to be normalized resulting in a non-linear relation.
\end{lemma}

\begin{proof}
For the ease of notations, without loss of generality, let $m_k$ be the selected variable node. 
 For the construction of the transfer matrix it is convenient to select the variable $m_k$ of constraint $m$. The input state now being $\chiket{m_1}\otimes \dots \otimes \chiket{m_{k-1}}$.
	Set $\chiket{m_j} = \alpha_j \zket + \beta_j \oket$.
	In order to satisfy the constraint, we want
	\begin{equation}\label{ksateqn}
		\braphi | \bc{\alpha_1\zket + \beta_1\oket}\tp\dots\tp\bc{\alpha_k\zket + \beta_k\oket}= 0.
	\end{equation}
 We expand this relation over the computational basis states $\vert i_1,\dots, i_k\rangle$, $i_j\in \{0, 1\}$.
	Defining 
 $$
 \gamma_j:=\begin{cases}
		\alpha_j \mbox{ if } i_j=0\\
		\beta_j \mbox{ if } i_j=1
	\end{cases}
 $$
	we can express \eqref{ksateqn} as follows
	\begin{equation}\label{prelin}
		\alpha_k\brk{\sum_{i_1,\dots,i_{k-1}\in\{0,1\}}\gamma_1\cdots\gamma_{k-1}\braphi|i_1\cdots i_{k-1}0\rangle} +\beta_k\brk{\sum_{i_1,\dots,i_{k-1}\in\{0,1\}}\gamma_1\cdots\gamma_{k-1}\braphi|i_1\cdots i_{k-1} 1\rangle} = 0.
	\end{equation}
	Therefore $\alpha_k$ and  $\beta_k$ can be found from the linear operation 	
	\begin{equation}\label{lin}
		\begin{bmatrix}
			\alpha_k\\
			\beta_k
		\end{bmatrix} \propto  \begin{bmatrix}
			\braphi|0\cdots 01\rangle & \cdots & \braphi|1\cdots 11\rangle \\
			- \braphi|0\cdots 00\rangle & \cdots & - \braphi|1\cdots 10\rangle
		\end{bmatrix} \bc{\begin{bmatrix}
				\alpha_1\\ \beta_1 
			\end{bmatrix} \tp \cdots \tp 
			\begin{bmatrix}
				\alpha_{k-1} \\ \beta_{k-1}
		\end{bmatrix}},
	\end{equation}
 and the resulting state can be normalized afterwards. 
The transfer matrix $T$ has size $2\times 2^{k-1}$ and is given by
\begin{equation}
T=	\begin{bmatrix}
	\braphi|0\cdots 01\rangle & \cdots & \braphi|1\cdots 11\rangle \\
	- \braphi|0\cdots 00\rangle & \cdots & - \braphi|1\cdots 10\rangle
\end{bmatrix}
\end{equation}
where the first (resp. second) row contains all $2^{k-1}$ `binary sequences' of the form $\vert i_1,\dots, i_{k-1}, 1\rangle$ (resp. $\vert i_1,\dots, i_{k-1}, 0\rangle$).
\end{proof}

We are now ready to explain the Algorithm (given in Table \ref{algo_rec_empty}) behind the proof of \ref{lr-pre-lemma}.
Let $D$ be the set of pairs $\cbc{v,a_v}$ of variables $v$ of degree one  removed in LR together with its unique adjacent clause $a_v$. We order $D$ chronologically, i.e., we let $$ D=\cbc{ \cbc{v_1,a_{v_1}},\cbc{v_2,a_{v_2}}, \ldots, \cbc{v_L,a_{v_L}}    }$$ where $v_i$ is the $i$-th removed degree-one variables. For $\alpha<\alphalr(k)$, LR ends \whp~with an empty 2-core and  
Algorithm \ref{algo_rec_empty} is a `reconstruction procedure' which yields a product state solution $\ket{\Psi}$. Without loss of generality, we can assume that the initial graph $G$ is connected, i.e., $G$ has no isolated variable nodes (as we can always assign an arbitrary qubit state to isolated variable nodes if they are present). In a nutshell, starting from the last deleted check node, Algorithm \ref{algo_rec_empty} recursively assigns values to the set of variables connected to a clause using the transfer matrix $T$ of Lemma \ref{lemma_T}. 
We use the notation, $T^{a_{v}}$ for the matrix $T$ corresponding to the projector $\vert\Phi^{a_v}\rangle\langle\Phi^{a_v}\vert$ associated to clause $a_v$.
We note that when a variable node $w$ connected to $a_v$ is already revealed in step $3$ of the algorithm, we only reveal the edge connecting $w$ and $a_v$.

\begin{algorithm}
	\caption{Reconstruction Algorithm} \label{algo_rec_empty}
	\KwIn{The ordered set $D$ } 
	\KwOut{A product state   $\ket{\Psi}= \ket{\chi_1} \otimes  \cdots \otimes \ket{\chi_{N}} \in \CC^{\otimes N}$}
	\Begin{
		\For{ i=N {\bf{to}} 1}{
			Reveal all $k$ variables connected to $a_{v_i}$\;
			\For{each variable $w \neq v_i$}{
			\If{ $w$ is not assigned any qubit }
			{Assign an arbitrary qubit $\ket{ \chi_w}$ to $w$ \;}
			\Else
			{Set $\ket{ \chi_w}$ to be the qubit corresponding to $w$ \;}
		}
		    Set $\ket{ \chi_{v_i}}=T^{a_{v_i}} \cdot \prod^{\otimes}_{w \neq {v_i}} \ket{\chi_w}$\;
	}
		
			}	
\end{algorithm}

\begin{proof}[Proof of Lemma \ref{lr-pre-lemma}]
By Lemma \ref{lemma_core}, for $\alpha < \alphalr(k)$ the LR ends with an empty core \whp~so the set $D$ contains all of the constraints nodes.  This ensures that the outputted product state $\ket{\Psi}$ of Algorithm \ref{algo_rec_empty}  has the correct length $N$, i.e., every variable has been assigned a qubit. Moreover,  Lemma \ref{lemma_T} ensures that at each step of the algorithm, the  projector $\vert\Phi^m\rangle\langle \Phi^m\vert$ is satisfied. While at the beginning of Algorithm \ref{algo_rec_empty} all the variables connected to the last constraint node in $D$ have not been assigned any value yet, as the algorithm runs we need to make sure that the variable $v$ connected to the check $a_v$ has no qubits states attached to it yet (otherwise we would almost certainly get a contradiction when using the transfer matrix). This is indeed the case as shown in Claim \ref{claim_leaf} below. Thus $\ket{\Psi}$ is a valid PRODSAT solution of the $k$-QSAT instance. Finally, as $ M = \alpha N$ and there are at most $k$  variables that are not assigned values for steps 3 to 8 of Algorithm \ref{algo_rec_empty}, the complexity of Algorithm \ref{algo_rec_empty} is $O(N)$.
\end{proof}

\begin{claim}\label{claim_leaf}
	For $\cbc{v,a_v} \in D$, the variable $v$ is assigned a qubit only when the check $a_v$ is revealed.
\end{claim}

\begin{proof}
	Suppose  $\lr$ ends at a time $T \geq 0$ and suppose in the reconstruction Algorithm \ref{algo_rec_empty}, we reveal $a_v$ at some time $0 \leq t \leq T$. Hence, during the LR, $v$ has degree one at time $T-t$. Now, assume that $v$ would already be assigned a qubit state at time $t$. So, $v$ must be connected to another clause $b$ which was already revealed before $a_v$ (during reconstruction), say at a time $s<t$. Thus, $b$ must have been removed in LR at time $T-s>T-t$. Therefore, at time $T-t$, $v$ is connected to both $a_v$ and $b$ and has degree at least $2$, a contradiction.
\end{proof}

\begin{remark}
	For $\alpha \geq \alphalr(k)$, LR ends with a non-empty core. As explained in the main text, if there is a zero energy product state $\vert\Psi^\prime\rangle$ on the core, we can apply a reconstruction procedure similar to Algorithm \ref{algo_rec_empty} to recover the full product state $\vert\Psi\rangle$. Steps 5 and 6 of Algorithm \ref{algo_rec_empty} must be adapted so that when $w$ belongs to the core then $\vert \chi_w\rangle$ is assigned the corresponding factor in $\vert\Psi^\prime\rangle$ and the for loop in step 2 will run for a number $L<N$ steps as the set $D$ does not contains all the constraints. This process can be used for $\alphalr(k) \leq \alpha<\alphadc(k)$. 
\end{remark}

    \section{Multivariate complex analysis results}\label{appcomplex}

To prove Proposition \ref{lemma_existsol} we rely on results from complex analysis of multivariate functions. Before stating these results we need to define {\it isolated} and {\it simple} zeros of mappings $f: \CC^n \to \CC^n$.

\begin{definition}
    A vector $\vec{a}$ is an {\it isolated} zero if it is the only solution of $f(\vz) = \vec{0}$ in some small enough neighborhood of $\vec{a}$.
\end{definition}
 For example the map $f_1(z_1, z_2) = z_1^2(z_2 - 1)$, $f_2(z_1, z_2) = z_2^2(z_1 - 3)$ has two isolated zeros $(z_1, z_2) = (0, 0)$, $(z_1, z_2) = (3, 1)$. On the other hand the mapping $f_1(z_1, z_2) = z_1^2 + z_2^2$, $f_2(z_1, z_2) = z_1 + i z_2$ has a family of zeros $(z_1, z_2) =(\mu, i \mu)$ for any $\mu\in \CC$, and none of these are isolated.

\begin{definition}
    A zero $\vec{a}$ is {\it simple} if it is isolated and if Jacobian satisfies $\det J_f(\vec{a}) \neq 0$.
\end{definition}
 For the first example above one checks that $\det J_f(z_1, z_2) = 4z_1z_2(z_1-3)(z_2-1) - z_1^2z_2^2$. Thus $(0,0)$ is not simple and $(3, 1)$ is simple. For the second example the zeros are not isolated and thus they are not simple. At the same time, the determinant of Jacobian vanishes. This is not a coincidence as one can show that a necessary condition to have $\det J_f(\vec{a}) \neq 0$ is that $\vec{a}$ is isolated. This follows from the local inverse theorem \ref{localinverse} below. This means that in the definition of a simple zero we can in fact drop the condition of being isolated. 

\begin{proposition}\cite[Proposition 2.1]{juzakov_book}
	\label{prop:juzakov}
	Suppose that the mapping $\vz\mapsto f(\vz)$
 is holomorphic in a domain $D \subset \CC^n$. 
	Suppose the closure of a neighborhood $U_{\vec{a}}\subset D$ of a zero $\vec{a}$ of the mapping does not contain other zeros (so $\vec{a}$ is isolated). Then there exists $\varphi>0$ such that for almost all $\zeta \in B(\vec{0},\varphi)$ (w.r.t Lebesgue measure), the mapping $\vz\mapsto f(\vz) -\zeta$
	has only simple zeros in $U_{\vec{a}}$. The number of simple zeros depends neither on $\zeta$ nor on the choice of the neighborhood $U_{\vec{a}}$.  
\end{proposition} 

\begin{definition}
The {\it multiplicity} of the zero $\vec{a}$ in the multivariate mapping $\vz\mapsto f(\vz)$ is defined as the number of zeros in $U_{\vec a}$ of the perturbed mapping $\vz\mapsto f(\vz) -\zeta$.  
\end{definition}

\begin{proposition}\cite[Proposition 2.2]{juzakov_book} 
\label{Juz2}
	The multiplicity of a simple zero is equal to 1.
\end{proposition} 

Applying Propositions \ref{prop:juzakov}, \ref{Juz2} to the mapping $f_1(z_1, z_2) = z_1^2(z_2 - 1)$, $f_2(z_1, z_2) = z_2^2(z_1 - 3)$ we see that 
the simple zero $(3, 1)$ has multiplicity one since the perturbed mapping develops ``one branch" of simple zeros $(z_1, z_2)\approx (1+ \frac{\zeta_1}{9}, 3+ \zeta_2)$, whereas $(0,0)$ has higher multiplicity as many branches of simple zeros $(z_1, z_2)\approx (\pm i \zeta_1^{1/2}, \pm \frac{i}{\sqrt 3} \zeta_2^{1/2})$ appear.

In our application we need to show that $\vec{0}$ is an isolated zero of $F_{\rm sqr}$. For this we rely on the  following local inverse theorem:
\begin{theorem}\cite[Theorem 5.5]{laurent2010holomorphic}\label{localinverse}
	Consider an holomorphic map $f: U\subset \CC^n \mapsto f(U)\subset\CC^n$. Suppose $\vec{a}\in U$. Then $f$ is biholomorphic in some small enough neighborhood of $\vec{a}$ if and only if $\det J_f(\vec{a})\neq 0$. (Biholomorphic means that the  map $f: U\mapsto f(U)$ is a bijection and its inverse $f^{-1}: f(U)\mapsto U$ is also holomorphic.)
\end{theorem}

    \section{Definition and Properties of the Buchberger algorithm  \label{App:buchberger}}

To describe the Buchberger algorithm, we need to define operations on multivariate polynomials. We start with a multivariate polynomial division.  

\subsection{Multivariate Division Algorithm}

Let $f$ be a multivariate polynomials \\ in $K[X_1, \dots, X_n]$. 
We can write $f = \sum_\alpha a_\alpha x^\alpha$  with $\alpha = (\alpha_1,\dots,\alpha_n)$ in $  \ZZ_{\geq0}^n$ where $\alpha_i$ denotes the exponent of the $i$th variable such that $x^\alpha$ is the monomial $x^\alpha = X_1^{\alpha_1}\dots X_n^{\alpha_n}$ and the coefficients  $a_\alpha $ are in $K$. 

For univariate polynomials in $X$, the common ordering of the monomials is the degree ordering, 
\begin{equation}\label{eq:univariate_ordering}
	\dots >X^n>X^{n-1}> \dots > X > 1.
\end{equation}
This notion can be extended to multivariate monomials.
\begin{definition}[\cite{ideals}]
	A monomial ordering $>$ on $K[X_1, \dots, X_n]$ is a relation $>$ on the set of monomial $x^\alpha, \alpha \in \ZZ_{\geq0}^n$ satisfying 
	\begin{enumerate}
		\item $>$ is a total ordering on $\ZZ_{\geq0}^n$
		\item If $\alpha > \beta$ then for all $\gamma\in\ZZ_{\geq0}^n$, $\alpha + \gamma > \beta + \gamma$.
		\item $>$ is a well-ordering on $\ZZ_{\geq0}^n$. This means that every nonempty subset of $\ZZ_{\geq0}^n$ has a smallest element under $>$. 
	\end{enumerate}
\end{definition}

We also need some additional definitions to describe the multivariate division.
\begin{definition}
	Let $f,g$ be two multivariate polynomials in $K[X_1, \dots, X_n]$ and $<$ be a monomial ordering on $K[X_1, \dots, X_n]$.
	\begin{itemize}
		\item The  multidegree of $f$ is $\multideg(f) = \max(\alpha \in \mathbb{Z}^n_{\geq0} | a_\alpha \neq 0)$ (the maximum is taken with respect to $<$). 
		\item The leading coefficient of $f$ is $LC(f)=a_{\multideg(f)} \in K $.
		\item The leading monomial of $f$ is $LM( f ) = x^{\multideg( f )}$.
		\item The leading term of $f$ is $LT(f) = LC(f) \cdot LM(f)$.
		\item The least common multiple of the leading terms of $f$ and $g$ is denoted $LCM(f,g)$ and $$LCM(f,g) = a_{\multideg(f)} b_{\multideg(g)} x^{\gamma}$$ with $\gamma_i = \max(\multideg(f)_i,\multideg(g)_i)$.
		\item The $S$-polynomial of $f$ and $g$ is $$ S(f,g) = \frac{LCM(f,g)}{LT(f)}\cdot f - \frac{LCM(f,g)}{LT(g)}\cdot g.$$
	\end{itemize}
	
\end{definition}

\begin{theorem}[Multivariate Division Algorithm \ref{alg:div} in{ $K[X_1, \dots, X_n]$} \cite{ideals}]
	\label{th:mv_div}
	Let $>$ be a monomial ordering on $\ZZ^n_{\geq 0}$, and let $F = (f_1,\dots,f_s)$ be an ordered s-tuple of polynomials in $K[X_1, \dots, X_n]$. Then every $p \in K[X_1, \dots, X_n]$ can be written as
	\begin{equation}
		p = q_1 f_1 + \cdots + q_s f_s + r
	\end{equation}
	where $q_i, r \in K[X_1, \dots, X_n]$, and either $r = 0$ or $r$ is a linear combination, with coefficients in $K$, of monomials, none of which is divisible by any of $LT( f_1),\dots, LT( f_s)$. We call $r$ a remainder of $p$ on division by $F$. Furthermore, if $q_i f_i \neq 0$, then
	$$\multideg( f ) \geq \multideg(q_i f_i ).$$
	Algorithm  \ref{alg:div} makes it possible to compute such decomposition.
\end{theorem}

The $r$ and $q_i$ polynomials depend on the ordering of the $(f_i)$ and on the monomial ordering. They are not uniquely characterized with the condition that $r$ is not divisible by $LT( f_1),\dots, LT( f_s)$ 

\begin{algorithm}
	\caption{Multivariate Division Algorithm} \label{alg:div}
	\KwIn{$f_1,\dots,f_s, p \in K[X_1, \dots, X_n]$ } 
	\KwOut{$q_1,\dots,q_s,r \in K[X_1, \dots, X_n]$}
	\Begin{
		$q_1 \leftarrow 0 , \dots, q_s \leftarrow 0, r \leftarrow 0 $\;
		\While {$p \neq 0$}{
			division$\_$occurred is false;\Comment{$p + f_1q_1+ \dots +f_sq_s + r $ is a loop invariant}\\
			\While{$i<s$ and not division$\_$occurred}{
				\If {$LT(f_i)$ divides $LT(p)$}{ 
					$q_i \leftarrow q_i + LT(p)/LT(f_i)$\;
					$p \leftarrow p - LT(p)/LT(f_i)\cdot f_i; $\Comment{Remove the leading monomial}
					division$\_$occurred is true\;
				}	
			}	
			\If {not division$\_$occurred}{
				$r \leftarrow r + LT(p)$\;
				$p \leftarrow p -  LT(p)$\;
			}
		}
	}	
\end{algorithm}

\begin{example}[Multivariate Division]\label{ex:mv_div}
	Let's divide $p = xy^2 + 1$ by $(f_1 = xy + 1, f_2 = y+1)$ using lexicographic ordering with $x > y$. 
	\begin{enumerate}
		\item  $LT(p)=xy^2$ which is divisible by $LT(f_1)=xy$. Then $q_1$ is updated to $q_1 = y $ and $p$ is updated to $p = xy^2+1 - xy^2 -y = -y +1$.
		\item  $LT(p) = -y$  which is only divisible by $LT(f_2) = y$. Then $q_2$ is updated to $q_2 = -1$ and $p$ is updated to $p = -y +1 + y + 1 = 2$.
		\item $LT(p) = 2 $ which is neither divisible by $LT(f_1)$ nor $LT(f_2)$ so $r$ is updated to $r=2 $ and $f= 0$. The algorithm ends with $q_1 = y,q_2  =1, r= 2.$
	\end{enumerate}
	If the division is performed by changing the order of the polynomial $(f_1 = y + 1, f_2 = xy+1)$, it gives another reminder: 
	\begin{enumerate}
		\item  $LT(p)=xy^2$ which is divisible by $LT(f_1)=y$. Then $q_1$ is updated to $q_1 = xy $ and $p$ is updated to $p = xy^2+1 - xy^2 -xy = -xy +1$.
		\item  $LT(p) = -xy$  which is divisible by $LT(f_1)$. Then $q_1$ is updated to $q_1 = xy-x$ and $p$ is updated to $p = -xy +1 + xy + x = x+1$.
		\item $LT(p) = x $ which is neither divisible by $LT(f_1)$ nor $LT(f_2)$ so $r$ is updated to $r=x+1 $ and $f= 0$. The algorithm ends with $q_1 = xy-x,q_2  =0,r= x+1.$
	\end{enumerate}
	
\end{example}

\subsection{Gröbner basis and Buchberger Algorithm}

An \textit{ideal} $I$ of a ring $R$ is an additive subgroup of the ring such that for all $x \in I, r \in R$, we have $rx \in I$.
The smallest ideal generated by a set $S$ of elements in $R$ is denoted as $\langle S \rangle = \{rx |  r \in R \text{ and } x \in S\}$. $S$ is  called a basis of $\langle S \rangle$.

\begin{definition}[Gröbner basis]
	Given an ideal $I$ of $K[X_1,\dots, X_n]$, $<$ a monomial ordering for $K[X_1,\dots, X_n]$ and a finite subset $G \subset I$, we say $G$ is a Gröbner basis for the ordering $<$ if
	\begin{equation*}
		\langle LT(G)\rangle = \langle LT(I)\rangle .
	\end{equation*}
\end{definition}
Here $\langle LT(G)\rangle$ means the ideal generated by the leading terms of the polynomials in $G$.

\begin{corollary}
	For a fixed monomial order, every ideal $I$ has a Gröbner basis. Furthermore, any Gröbner basis is a generating set of $I$.
\end{corollary}
Gröbner bases have nice properties. Running the multivariate division algorithm with a Gröbner basis will not change the remainder regardless of the chosen order of the polynomials in the basis.
Indeed, if the Buchberger outputs two different remainders $r$ and $r'$ for the division of a polynomial $f$ with a Gröbner basis $G = \{g_1,\dots,g_s\}$, then there exist $g,g' \in <G>$ such that $f = g + r = g' + r'$. Thus $r-r'= g'-g \in <G>$ so $LT(r-r') \in <LT(g_1),\dots,LT(g_s)>$ by the definition of a Gröbner basis but from Theorem \ref{th:mv_div} neither $r$ nor $r'$ has monomials divisible by any of the $LT(g_1),\dots,LT(g_s)$. This implies that $r=r'$.

The Buchberger algorithm (\ref{algo_buchberger}) returns a Gröbner basis.
\begin{algorithm}[h]
	\caption{Buchberger Algorithm \cite{ideals}} \label{algo_buchberger}
	\KwIn{$F = (f_1,\dots,f_s)$} 
	\KwOut{A Gröbner basis $G$}
	\Begin{
		$G \leftarrow F$ \;
		\Repeat {$G = G'$}{
			$G'  \leftarrow G$\;
			\For{ each pair $g_i,g_j \in G$ with $ i\neq j $}{
				$S_{i,j} \leftarrow S(g_i,g_j)$\;
				$r  \leftarrow$Multivariate Division Algorithm $(G,S_{i,j})$\;
				\If{$r \neq 0$}{
					$G  \leftarrow G \cup \{ r\}$\;
				}
			}
		}		
	}	
\end{algorithm}

\begin{example}
	Let us construct the Gröbner basis for the two polynomials of Example \ref{ex:mv_div}, $(f_1 = xy + 1, f_2 = y+1)$ with lexicographic ordering. 
	\begin{enumerate}
		\item $S(f_1,f_2) = xy +x -xy  - 1 = x -1$ whose leading term is divisible neither by $LT(f_1)$ nor $LT(f_2)$. So, $f_3 = x - 1$ can be added to the basis. 
		\item $S(f_1,f_3) = xy +1 -xy  +y  = y+1 = f_2 $. The remainder of the division is 0. 
		\item $S(f_2,f_3) = xy +x -xy  +y  = x - y = f_1 + f_3$. The remainder of the division is 0. All pairs have been examined and the algorithm ends.
	\end{enumerate}
	One can verify that running the algorithm with $(f_1,f_2,f_3)$ yields the same remainder regardless of the order of the polynomials.
\end{example}

Buchberger algorithm generates numerous intermediate polynomials with total degrees that can be pretty large. 
The basis output by Buchberger algorithm can be simplified.

\begin{definition}[Reduced Gröbner basis \cite{ideals}]
	For an ideal $I \subseteq K[X_1,\dots, X_n]$, a finite subset $G \subset I$ is a reduced Gröbner basis of $I$ for the order $<$ if 
	\begin{itemize}
		\item[-] $ LC(g) = 1$ for all $g \in G$,
		\item[-] for all $g \in G$ no monomials of g lie in $\langle LT(G \setminus \{g\})\rangle$.
	\end{itemize}
\end{definition}
\begin{theorem}[\cite{ideals}]
	Let $I \neq \cbc{0}$ be a polynomial ideal. Then, for a given monomial ordering, $I$ has a reduced Gröbner basis and the reduced Gröbner basis is unique.
\end{theorem}

\begin{example}
	$(f_2,f_3)$ is the reduced basis for $(f_1,f_2,f_3)$.
\end{example}

We can construct a reduced Gröbner basis for a non-zero ideal by applying Buchberger algorithm. Then by adjusting the constants of the obtained basis $G$ to make all the leading coefficients equal to $1$ and removing any $g$ with $LT(g) \in \langle LT(G \setminus \{g\})\rangle$ from $G$, we can obtain the reduced basis because for any removed $g$, the resulting set $G \setminus \{g\}$ is also a Gröbner basis.

Even using a reduced Gröbner basis, Buchberger algorithm takes a huge amount of storage. The degree of the polynomials in the reduced Gröbner basis is bounded by $2(d^2/2+d)^{2^{n-2}}$ \cite{dube} where $d$ is the total degree of the $(f_i)$ for example $d= k$ for $k$-QSAT. 
Today Faugere's algorithms are the fastest algorithms to compute Gröbner basis \cite{F4,F5}. They use the same principle as Buchberger algorithm but use linear algebra to evaluate several pairs of polynomials and apply additional criteria to avoid evaluating $S$-polynomials that will reduce to $0$.

\subsection{Hilbert’s Nullstellensatz}

Thanks to Hilbert’s Nullstellensatz, Buchberger algorithm can determine if a system of complex multivariate polynomial equations admits a common zero. Let us recall that an ideal $I$ is a maximal ideal of a ring $R$ if there are no other ideals contained between $I$ and $R$, i.e. for any ideal $J$ such that $I \subsetneq J, J = R$ .

\begin{theorem}[Hilbert’s Nullstellensatz (Zeros Theorem)]\label{th:zeros}
	Let $K$ be an algebraically closed field. Then every maximal ideal in the polynomial ring $K[X_1,\dots,X_n]$ has the form $(X_1 - a_1,\dots,X_n -a_n)$ for some $a_1,\dots, a_n \in K$.
\end{theorem}

\begin{corollary}
	As a consequence, a family of polynomials in $K[X_1,\dots, X_n]$ with no common zeros generates the unit ideal.
\end{corollary}
\begin{proof}
	Let $I$ be an ideal generated by  $f_1, \dots f_m \in K[X_1,\dots,X_n]$ with no common zeros. 
	If $I$ is contained in a maximal ideal $ M = (X_1 - a_1,\dots, X_n - a_n)$ by Theorem \ref{th:zeros}, then $(a_1, \dots, a_n) \in K^n$ is a common root of elements of $I$, in contradiction with the hypothesis. 
	Since $I$ does not lie in any maximal ideal, it must be $K[X_1,\dots, X_n]$. 
\end{proof}
\begin{remark}
	Conversely, any family of polynomials $f_1, \dots f_m$ in $K[X_1,\dots, X_n]$ that generates the unit ideal has no common zeros. Indeed if $\langle f_1, \dots f_m \rangle = K[X_1,\dots,X_n]$ there exist $g_1,\dots,g_m \in K[X_1,\dots,X_n]$ such that $g_1f_1+ \cdots + g_mf_m = 1$. If $f_1, \dots f_m$ have a common zero, it contradicts the equality.
\end{remark}

\begin{corollary}\label{cor:null}
	A set of polynomials in an algebraically closed field has no common zeros  if and only if the reduced Gröbner basis is $\{1\}$.
\end{corollary}
\begin{proof}
	A set of polynomials has no common zeros if and only if it generates the unit ideal and $1$ belongs to an ideal if and only if $1$ belongs to the Gröbner basis of the ideal for any monomial ordering (because $LT(1) = 1$) and thus belongs to the reduced Gröbner basis.
\end{proof}

From Corollary \ref{cor:null}, the Gröbner basis output by Buchberger algorithm on input $f_1,\dots,f_m$ is reduced to $\{1\}$ if and only if  $f_1,\dots,f_m$ do not have a common zero.

    \section{Evaluating the dimension of the kernel }\label{app:pattern}

In Table \ref{tab:pattern}, we give recurrence relations for $\dimKerH$ with a specific interaction graph. Here we give the proof of the recurrence relations for the loose chain and the cycle (when the projectors are separable) which are not found in the literature to the best of our knowledge. These recurrences in Lemmas \ref{lem:loose_chain_sep} and \ref{lem:loose_cycle_sep} yield upper bounds since we prove them only for separable projectors. However, we observe with numerical tests that they coincide with the `generic' values.

\begin{lemma}
    \label{lem:loose_chain_sep}
    For an instance of $k$-QSAT with $m$ separable projectors, the dimension of the kernel space for the loose chain interaction graphs, satisfies the recurrence relation 
    \begin{equation}
    \label{rec_chain_loose}
        r_m = 2^{k-1} r_{m-1} - 2^{k-2}r_{m-2}.
    \end{equation} 
    The initial conditions are $r_1=7$, $r_2=24$. 
\end{lemma}

The initial conditions for $m = 1,2$ are found by considering a special case of the $\vec{d}$-nosegay described in \cite{bounds} where $\vec{d}=(0,\dots,0)$ for $m = 1$ and $\vec{d}= (1,0,\dots,0)$ for $m = 2$. The proof of the recurrence relation is based on the arguments similar to those in \cite[Lemma 5]{bounds}.

\begin{proof}
    If the projectors are separable, then we can find a basis of the Hilbert space to decompose the two projectors at the ends of the chain as $ \ket{\alpha_m} \otimes \zket^{\otimes k-1}$ and the interior projectors as $\ket{\beta^0_m} \otimes \ket{\beta^1_m} \otimes \zket^{\otimes k-2}$ where $\ket{\alpha_m}$ and $\ket{\beta^i_m}, i\in \{0,1\}$ are the states of the qubits that appear in two clauses.

    We can construct  a basis for the solution space of the form 
    \begin{equation}
        \ket{\vec{b}} = \bigotimes_i \ket{b_i} \otimes \ket{v_{\vec{b}}}
    \end{equation}
    where $\ket{b_i}$ is the state of a unary qubit (a qubit whose vertex is unary) and  $\ket{v_{\vec{b}}}$ is the state of the remaining qubits. The solutions are constructed by satisfying the clauses from one end to the other of the loose chain. 
    
    The first clause, at the beginning of the chain, can be satisfied if one of its $k-1$ unary qubits has a state equal to $\oket$. In this situation, there are  $2^{k-1}-1$ possible ways to satisfy the first clause. The remaining degree-two qubit is left unassigned. What remains to satisfy is a loose chain with $m-1$ clauses. The subspace satisfying the new pattern is of dimension $r_{m-1}$.
    
    If the states of all the unary qubits of the first clause are set to $\zket^{\otimes k-1}$, then the last qubit is constrained to satisfy the clause. 
    The Schmidt decomposition of $\ket{v_{\vec{b}}}$ gives  $\ket{v_{\vec{b}}} = \ket{q_1} \otimes \ket{v_1} + \ket{q_2} \otimes \ket{v_2}$. Since $\ket{q_1}$ is orthogonal to $\ket{q_2}$ and $ \langle q_1| \alpha_1\rangle = 0 = \langle q_2 | \alpha_1 \rangle$ (because the first clause is satisfied), $\ket{v_{\vec{b}}}$ is separable and $\ket{q_1} = \ket{q_2} = \ket{\alpha_1^\perp}$. To extend this partial assignment to a full solution, we can apply the same method recursively for the interior clauses where there are only $k-2$ unary qubits. 
    For the last exterior clause, the degree 2 qubit is fixed so the dimension is $2^{k-1} - 1$.
    This gives in total 
    \begin{equation}\label{intotal}
       r_m = (2^{k-1}-1) r_{m-1} + (2^{k-2}-1) \sum_{i = 2}^{m-2} r_{i} + 2^{k-1} - 1.
    \end{equation}

    To obtain the announced relation 
    \eqref{rec_chain_loose}, a few algebraic manipulations are necessary. Applying the last result to $r_{m-1}$, we obtain
    \begin{align}
       r_{m-1} - (2^{k-1}-1) r_{m-2} & = (2^{k-2}-1) \sum_{i = 2}^{m-3} r_{i} + 2^{k-1} - 1 
    \end{align}
    and combining with \eqref{intotal}, we find
    \begin{align}
        r_m &= (2^{k-1}-1) r_{m-1} + (2^{k-2}-1)r_{m-2}  + r_{m-1} - (2^{k-1}-1) r_{m-2}
        \nonumber\\
        &= 2^{k-1} r_{m-1} + (2^{k-2} - 2^{k-1})r_{m-2}
        \nonumber\\
        &
        = 2^{k-1} r_{m-1} - 2^{k-2}r_{m-2}.
    \end{align}
\end{proof}

\begin{lemma}\label{lem:loose_cycle_sep}
    For an instance of $k$-QSAT with $m$ separable projectors, the dimension of the kernel space for the loose cycle interaction graphs, satisfies the recurrence relation 
    \begin{equation}\label{rec_cycle_loose}
        s_m = 2^{k-1} s_{m-1} - 2^{k-2}s_{m-2}.
    \end{equation} 
    The initial conditions are $s_2=12$, $s_3=40$.
\end{lemma}

It is more convenient to use  the notation $s_m$ (instead of $r_m$ used in Table \ref{tab:pattern})    because the proof will use the previous lemma. We also need the following lemma.

\begin{lemma}\label{lem:3_then_2}
For an instance of QSAT with $m + 1$ separable projectors whose interaction graph is a 2-qubit chain starting with a $k$-qubit clause and $k\geq 2$, the dimension of the kernel space satisfies the recurrence relation
 \begin{equation}\label{rec_tree_loose}
        t_m =  (2^{k-1} - 1)m + 2^{k} - 1.
    \end{equation} 
Here $m$ is the number of $2$-qubit clauses.
\end{lemma}

\begin{proof}
    Recall that the dimension of the kernel space for a chain of $m$ $2$-qubit clauses is $m+2$ \cite{QSATLaumann}. 
    The qubits of the first clause are labeled from $1$ to $k$ starting with the qubit both in the chain and in the $k$-qubit clause.
    We can find a basis of the Hilbert space where the projector of the first clause can be decomposed into $\ket{\phi} \otimes \ket{0}^{\otimes k-2}$ where $\ket{\phi}$ is a state of the first two qubits. 
    We can construct a basis for the solution space of the form 
    \begin{equation}
        \ket{\vec{b}} = \bigotimes_{i=3}^{k} \ket{b_i} \otimes \ket{v_{\vec{b}}}
    \end{equation}
    where $\ket{b_i}$ is the state of qubit $q_i$ and $\ket{v_{\vec{b}}}$ is the state of the remaining qubits. 
    If any of the $\ket{b_i}$ are $\oket$, the first clause is satisfied. What remains to satisfy is a $2$-qubit chain of $m$ clauses. The unassigned qubit $q_2$ accounts for 2  degrees of freedom.  There are $2^{k-2}-1$ possibilities to satisfy the first clause this way. 
    If the qubits $q_2,\dots,q_k$ are assigned to $\ket{0}^{\otimes k-2}$, then it remains to satisfy a $2$-qubit chain of $m+1$ clauses. Summing all contributions we find 
    \begin{equation}
        t_m = 2(2^{k-2}-1)(m+2) + (m+3) = (2^{k-1}-1) m + 2^{k} - 1.
    \end{equation}
\end{proof}

\begin{proof}[Proof of Lemma \ref{lem:loose_cycle_sep}] 
    We start by looking at the interaction graphs that are loose chains composed of $m$ clauses of $k$-qubit and ending with $p$ $2$-qubit clauses.  Let us denote $r_{m,p}$ the dimension of the kernel space of the instances with these interaction graphs.
    We can show the following recurrence relation over $m$
    \begin{equation} \label{eq:mixed}
        r_{m,p} = 2^{k-1}r_{m-1,p} - 2^{k-2} r_{m-2,p}
    \end{equation}
    with the same argument of the proof of Lemma \ref{lem:loose_chain_sep} and with the initialization given by $r_{m,0} = m+2$ (chain only composed of $2$-qubit clauses) and $r_{m,1} = t_m$ from Lemma \ref{lem:3_then_2}.

    Regarding the loose cycle, we remark that fixing the value of a unary qubit in a clause $c$ breaks the cycle into a loose chain with $m-1$ clauses if this assignment satisfies $c$. There are $2^{k-2} -1$ ways to satisfy a clause with unary qubits. If, after assigning a value to all unary qubits in $c$, the clause is still unsatisfied, the resulting interaction graph is a loose cycle with one 2-qubit clause. We can repeat the procedure and assign unary qubits of the next clause $c+1$ in the cycle. If clause $c+1$ is satisfied, the new interaction graph is a loose chain ending with one $2$-qubit clause. If $c+1$ is unsatisfied, the cycle graph now contains two $2$-qubit clauses. We can iterate this procedure. At step $1 \leq p\leq m$, if the clause $c+p$ is satisfied, the cycle is broken into a loose chain ending with $p-1$ $2$-qubit clauses and if the clause $c+p$ is unsatisfied the new interaction graph is a loose cycle with $p$ $2$-qubit clauses. 

    After $m$ steps, if all unary qubits are assigned and do not satisfy any of the clauses, the interaction graph is a cycle composed of only $2$-qubit clauses, and the dimension of the kernel space for this graph is $2$ \cite{QSATLaumann}. Summing all contributions gives 
    \begin{equation}
        s_m =  (2^{k-2} -1 ) \sum_{i=0}^{m-1}  r_{i,m-1-i} + 2.
    \end{equation}
    With some algebraic manipulations, we can obtain the desired relation, i.e,
    \begin{align}
        \nonumber
	    s_m &= (2^{k-2} -1) \sum_{i=2}^{m-1}  r_{i,m-1-i}  +  (2^{k-2} -1 ) ( r_{0,m-1} + r_{1,m-2}) + 2\\ \label{eq:37}
         &= (2^{k-2} -1) \sum_{i=2}^{m-1}  \left( 2^{k-1} r_{i-1,m-1-i} - 2^{k-2} r_{i-2,m-1-i} \right ) + (2^{k-2} -1 ) ( r_{0,m-1} + r_{1,m-2}) + 2\\ \label{eq:38}
	    &= 2^{k-1} s_{m-1} - 2^{k-2} s_{m-2}  + (2^{k-2} -1 )( r_{0,m-1}  + r_{1,m-2} - 2^{k-1} r_{0,m-2})  - 2^{k-1} + 2  \\
	    &= 2^{k-1} s_{m-1} - 2^{k-2} s_{m-2}. \label{eq:39}
    \end{align}
 Equation (\ref{eq:37}) is obtained by applying (\ref{eq:mixed}) to each $r_{i,m-1-i}$. In (\ref{eq:38}), we bring out $s_{m-1}$ and $s_{m-2}$. Finally, we replace in (\ref{eq:38}) $r_{0,m-1},r_{0,m-2}$ and $r_{1,m-2}$ by their algebraic expressions to obtain (\ref{eq:39}).
\end{proof}

\end{appendix}

\newpage

\printbibliography

\end{document}